\documentclass[11pt]{article}

\usepackage{amssymb} 
\usepackage{amsthm,amsmath,amssymb,bbm}

\usepackage{latexsym}
\usepackage{colordvi,xcolor}
\usepackage{enumerate}
\usepackage[toc,page]{appendix}
\usepackage{listings}
\usepackage{extarrows}
\usepackage{graphicx}
\usepackage{float}
\usepackage{enumitem}
\usepackage{multirow}
\usepackage{threeparttable}
\usepackage{natbib}
\usepackage{ulem}
\usepackage{colortbl}

\usepackage{colordvi}
\usepackage{xcolor}
\usepackage{color, colortbl}
\usepackage{color,soul}
\usepackage{setspace}
\usepackage{multirow}
\usepackage{multicol}
\usepackage{booktabs} 
\usepackage{bm}
\usepackage{textcomp}
\usepackage{rotating}
\usepackage{epsfig}
\usepackage{subfigure} 
\usepackage{footnote}
\usepackage{url}
\usepackage{listings}
\usepackage{fancybox}
\usepackage{lscape}
\usepackage{comment}
\usepackage{latexsym}
\usepackage{colordvi}
\usepackage{enumerate}
\usepackage[toc,page]{appendix}
\usepackage{listings}
\usepackage{extarrows}
\usepackage{float}
\usepackage{enumitem}
\usepackage{booktabs}
\makeatother
\usepackage{verbatim}   
\usepackage{color}      
\usepackage{enumerate}
\usepackage{threeparttable}
\usepackage{latexsym}
\usepackage{colordvi}
\usepackage{listings}
\usepackage{extarrows}
\usepackage{float}
\usepackage{enumitem}
\usepackage{multirow}
\usepackage{tablefootnote}
\makeatother
\newtheorem{assumption}{Assumption}
\newtheorem{prop}{Proposition}
\newtheorem{lemma}{Lemma}
\newtheorem{theorem}{Theorem}
\newtheorem{remark}{Remark}

\newcommand{\bfeta}{\boldsymbol{\eta}}

\newcommand{\bgamma}{\boldsymbol{\gamma}}

\newcommand{\bbeta}{\boldsymbol{\beta}}

\newcommand{\bzeta}{\boldsymbol{\zeta}}

\newcommand{\bOmega}{\boldsymbol{\Omega}}

\newcommand{\bPhi}{\boldsymbol{\Phi}}

\newcommand{\bPsi}{\boldsymbol{\Psi}}

\newcommand{\bZ}{\boldsymbol{Z}}
\newcommand{\bX}{\boldsymbol{X}}
\newcommand{\bx}{\boldsymbol{x}}

\newcommand{\bb}{\boldsymbol{b}}
\newcommand{\bW}{\boldsymbol{W}}

\newcommand{\bR}{\boldsymbol{R}}

\newcommand{\bJ}{\boldsymbol{J}}

\newcommand{\bH}{\boldsymbol{H}}
\newcommand{\bfm}{\boldsymbol{m}}

\newcommand{\bU}{\boldsymbol{U}}
\newcommand{\bu}{\boldsymbol{u}}

\newcommand{\bv}{\boldsymbol{v}}

\newcommand{\bzero}{\boldsymbol{0}}

\def\ci{\perp\!\!\!\perp}

\newcommand{\Rom}[1]{\text{\uppercase\expandafter{\romannumeral #1\relax}}}

\usepackage{geometry}
\usepackage{enumitem}

\numberwithin{equation}{section}

\begin{document}

\title{Estimation of Complier Expected Shortfall Treatment Effects with a Binary Instrumental Variable}

\author{Bo Wei,~Kean Ming Tan,~and~Xuming He \\ University of Michigan}

\date{}
\maketitle

\begin{abstract}
Estimating the causal effect of a treatment or exposure for a subpopulation is of great interest in many biomedical and economical studies.  Expected shortfall, also referred to as the super-quantile, is an attractive effect-size  measure that can accommodate data heterogeneity and aggregate local information of effect over a certain region of interest of the outcome distribution. In this article, we propose the \underline{C}omplie\underline{R} \underline{E}xpected \underline{S}hortfall  \underline{T}reatment \underline{E}ffect (CRESTE) model under an instrumental variable framework  to quantity the CRESTE for a binary endogenous treatment variable. By utilizing the special characteristics of a binary instrumental variable and a specific formulation of Neyman-orthogonalization, we propose a  two-step estimation procedure, which can be implemented by simply solving weighted least-squares regression and weighted quantile regression with estimated weights. We develop the asymptotic properties for the proposed estimator and use  numerical simulations to confirm its validity and robust finite-sample performance. An illustrative analysis of a  National Job Training Partnership Act study is presented to show the practical utility of the proposed method.
\end{abstract}

\section{Introduction\label{intro}}
Given a scalar response and a set of covariates, it is often of interest to understand the effect of a treatment or an intervention on the conditional distribution of the response. 
Compared to the popularly used least squares approach, quantile regression-based approach is more appealing due to its ability to capture heterogeneous treatment effects, i.e., the treatment effect that  may vary across different regions of the conditional distribution of the response even after adjusting for the observed covariates \citep{Koenker2005}. Several approaches have been proposed to estimate quantile treatment effects \cite[among others]{Abadie2002, Chernozhukov2005}. We refer the reader to \cite{Koenker2017} for an extensive review of quantile regression and quantile treatment effects.

In many applications, however, the average treatment effect in one tail of the response distribution may be of great interest in contrast to the effect at a specific quantile level \citep{Acharya2017,Brownlees2016}.    For example, the average effect of  public
sector-sponsored training programs on the low-income subpopulation may be particularly important for making policy decisions \citep{Lalonde1995}.  To this end, the expected shortfall, also known as the super-quantile or conditional value-at-risk, has proven to be useful.

Various methods for estimating the expected shortfall were discussed in \cite{Taylor2008a,Taylor2008b, Brazauskas2008, Cai2008, Chen2008} and \cite{Rockafellar2013}. One can also integrate over conditional quantile regression estimates in the tail \citep{WW16}.
More recently, several authors have considered expected shortfall of the  response variable conditional on a set of covariates under the regression setting \citep{Dimitriadis2019,Barendse2020,he2022robust}. 
Specifically, \cite{Dimitriadis2019} proposed joint quantile and expected shortfall regression models to estimate the expected shortfall regression coefficients with a loss function, referred to as the FZ-loss  \citep{Fissler2016}. On the other hand, \cite{Barendse2020} proposed a two-stage estimation method based on Neyman-orthogonalization for estimating the expected shortfall regression coefficients.  Motivated by \citet{Barendse2020}, \citet{he2022robust} proposed the robust expected shortfall regression that is robust against heavy-tailed random noise. 

In this article, we are interested in estimating the treatment effect on the expected shortfall of the outcome distribution, adjusted for a set of covariates, in observational studies. 
One motivating example is the National Job Training Partnership Act (JTPA) study, a large publicly-funded training program, also considered in \citet{Abadie2002} in the context of estimating quantile treatment effect.  
The main goal is to evaluate the effect of JTPA training on improving 30-months earnings for low-income individuals. 
However, as pointed out by \citet{Abadie2002}, the training status in the JTPA study is likely to be self-selected and is correlated with the potential outcome of earnings. 
Using standard approaches for estimating expected shortfall effect without accounting for the selection bias would lead to an invalid estimate of the expected shortfall treatment effect.

To address the above challenge, instead of estimating the overall expected shortfall treatment effect, we propose to estimate the expected shortfall treatment effect for subjects who comply with the treatment protocol, referred to as compliers.   The complier expected shortfall treatment effect (CRESTE) represents a local average effect on the tail of the distribution, rather than a specific quantile level. Compared to the complier quantile treatment effect (CQTE) proposed in \citet{Abadie2002}, CRESTE is a causal estimand that is more suitable to characterize the behavior of the tail treatment effect. This is because, analogous to comparing the quantile and the expected shortfall, CQTE, defined based on the quantile of the response distribution, fails to capture the tail behavior beyond the quantile itself, while CRESTE characterizes the tail behavior by aggregating information from the entire tail region \citep{Rockafellar2002}.

To the best of our knowledge, there is limited work on the estimation of CRESTE in the existing literature, except a concurrent work by \citet{Chen2021} who proposed to estimate the CRESTE based on the method for fitting expected shortfall regression in \citet{Dimitriadis2019}.  
The loss function used in \citet{Chen2021} inherits the non-convexity of the loss function in \cite{Dimitriadis2019} for which the global optimum solution is not guaranteed.  
Moreover, their proposed CRESTE estimate is not locally robust to the complier quantile treatment effect estimate in the sense of robustness in \cite{Chernozhukov2016}, and a convergence rate $O(n^{-1/2})$ for the CQTE estimate is required to ensure asymptotic normality for the CRESTE estimator.

In this manuscript, we derive a weighting scheme that can be incorporated into the two-step estimation procedure in \cite{Barendse2020} by utilizing the special characteristic of a binary instrumental variable to estimate CRESTE. The validity of the proposed method relies only on the modeling of the complier subgroup and does not require the modeling of the other non-complier subgroups or the instrumental variable distribution. The proposed two-stage estimation procedure involves fitting a  weighted quantile regression at the first stage, and fitting a weighted least squares model at the second stage.  Of practical appeal is that both steps involve minimizing convex objective functions for which global optimum are guaranteed, and can be implemented by existing software for fitting quantile regression and the least squares regression.
Moreover, the proposed method has the property of Neyman-orthogonalization for the estimation of the expected shortfall effect.  The implication is that the resulting estimator is locally robust to the estimate of quantile, which allows flexible approaches to be used for modeling quantile regression in the first stage, such as nonparamaetric quantile regression \citep{He1994, De2003}. 

\section{Preliminaries on Expected Shortfall Regression and Potential Outcomes Framework}

In this section, we provide an overview of the expected shortfall regression and the potential outcomes framework that are essential in the development of our proposed method.

\subsection{Expected Shortfall Regression}
The expected shortfall of a continuous random variable $Z$ is the conditional expectation of $Z$, conditioned on $Z$ falling below a given quantile level of its distribution. 
Specifically, the expected shortfall of $Z$ at level $\alpha\in(0,1)$ is defined as  
\[
S_{\alpha}(Z):=\mathbb{E}\{Z | Z\le Q_{\alpha}(Z)\}=\frac{1}{\alpha}\int_{0}^{\alpha}Q_{u}(Z)du,
\]
where $Q_{u}(Z)=\inf\{z\in\mathbb{R}:\Pr(Z\le z)\ge u\}$ is the $u$th quantile of $Z$ \citep{Yamai2002}. 
Several authors have generalized the expected shortfall to the regression setting to evaluate the association between a covariate of interest and the response, given a set of nuisance covariates; see among others, \cite{He2010,Dimitriadis2019,  Barendse2020,he2022robust}.

Let $Y$ be the response variable and let $\bX\in\mathbb{R}^l$ be the covariates. 
Moreover, let $Q_{\alpha}(Y|\bX)$ and $S_{\alpha}(Y|\bX)$ be the $\alpha$th quantile and expected shortfall of $Y$ conditional on $\bX$, respectively. 
While the expected shortfall is not elicitable, i.e., there does not exist a loss function  that is minimized by the  expected shortfall regression parameters  \citep{gneiting2011making},   the quantile and the expected shortfall are jointly elicitable \citep{Fissler2016}.  
This motivates the following joint quantile and expected shortfall regression framework \citep{Dimitriadis2019,Barendse2020,he2022robust}:
\begin{equation}
\begin{split}
Q_{\alpha}(Y|\bX)=\bX^\top\bbeta^*(\alpha) \qquad \mathrm{and}\qquad
S_{\alpha}(Y|\bX)=\bX^\top\bgamma^*(\alpha),
\end{split}
\label{SQtwostep1}
\end{equation}
where $\bbeta^*(\alpha)$ and $\bgamma^*(\alpha)$ are the true underlying regression coefficients corresponding to the $\alpha$th quantile and expected shortfall regression, respectively. 
Under a two-step framework, \cite{Barendse2020} proposed to first estimate $\bbeta^*(\alpha)$ by fitting a quantile regression model.  Given a quantile regression estimator $\hat{\bbeta}(\alpha)$, the expected shortfall regression coefficients $\bgamma^*(\alpha)$ can then be estimated by solving the following convex optimization problem:
\begin{equation}
\hat{\bgamma}(\alpha)=\underset{\bgamma (\alpha) \in \mathbb{R}^l}{\mathrm{argmin}}~ \frac{1}{n}\sum_{i=1}^n\left[\frac{1}{\alpha}\{Y_i-\hat{Q}_{\alpha}(Y_i|\bX_i)\}I\{Y_i\le\hat{Q}_{\alpha}(Y_i|\bX_i)\}+\hat{Q}_{\alpha}(Y_i|\bX_i)-\bX_i^\top\bgamma(\alpha)\right]^2,
\label{SQtwostep2}
\end{equation}
where $I(\cdot)$ is an indicator function and  $\hat{Q}_{\alpha}(Y_i|\bX_i) = \bX_i^{\top} \hat{\bbeta}(\alpha)$.

Compared to the method in \cite{Dimitriadis2019}, the two-step framework has advantages in both computational and statistical aspects. 
Theoretically, it is remarkable that~\eqref{SQtwostep2} is Neyman-orthogonalized. 
Let $\|\cdot\|$ be the $\ell_2$-norm.
The implication of Neyman-orthogonalization in \cite{Barendse2020} is that the asymptotic distribution of the estimator for $\bgamma^*(\alpha)$  can be established under a weaker condition  $\|\hat{\bbeta}(\alpha)-\bbeta^*(\alpha)\|=o_p(n^{-1/4})$ \citep{Barendse2020} than that of \citet{Dimitriadis2019}, which requires $\|\hat{\bbeta}(\alpha)-\bbeta^*(\alpha)\|=O_p(n^{-1/2})$.  
 We note that the Neyman-orthogonalization phenomenon has been observed in \cite{Belloni2014,Chernozhukov2018,Chernozhukov2016} under different contexts.

The weaker condition on the quantile estimate allows flexible approaches such as the nonparametric quantile regression to be used. 
Computationally, the two-stage method involves fitting a weighted quantile regression and a weighed least squares regression that can be readily solved via existing software. 
Compared to the approach in \cite{Dimitriadis2019} that involves solving a non-convex optimization problem for which global optimum is not guaranteed, we found that the two-stage method has a more  stable numerical performance.


\subsection{The Potential Outcomes Framework}
\label{POFA}
In this section, we provide a brief overview of the potential outcomes framework for evaluating complier treatment effects.  In addition, we review the use of instrumental variable for estimating the complier quantile treatment effect \citep{Abadie2002}.  
  
Let $D$ and $V$ be indicators of a binary exposure and a binary instrumental variable, respectively. Define $D_v$ as the potential treatment selection given $V=v$. Moreover, let $Y$, $Y_d$, and $Y_{vd}$ be the observed outcome,  the potential outcome given $D=d$, and the potential outcome given $D=d$ and $V=v$, respectively. 
We note that in practice, $D_v$, $Y_d$, and $Y_{vd}$ are not observed. 

Under the above binary exposure and instrumental variable setting, subjects can be classified into four latent subgroups: {\it compliers} ($D_1>D_0$), {\it always takers} ($D_1=D_0=1$), {\it never takers} ($D_1=D_0=0$), and {\it defiers} ($D_1<D_0$) \citep{Angrist1996}. 
Let $\bX$ be an $l$-dimensional vector of covariates with one as the first component, without loss of generality. 
We start with some assumptions on the instrumental variable $V$. 
\begin{assumption}
The instrumental variable $V$ satisfies the following conditions: 

\begin{itemize}
\setlength{\itemindent}{0.3in}
\item[(A1)]{ Independence of IV:
$(Y_{00},Y_{01},Y_{10},Y_{11}, D_0, D_1)\ci V|\bX;
$
}
\item[(A2)]{
Exclusion of IV: $P(Y_{1d}=Y_{0d}|\bX)=1$ for $d=0,1$;
}
\item[(A3)]{First stage: $0<P(V=1|\bX)<1$ and $P(D_1=1|\bX)>P(D_0=1|\bX)$;

}
\item[(A4)]{ Monotonicity: $P(D_1\ge D_0|\bX)=1$.
}
\end{itemize}
\label{assumption:iv}
\end{assumption}
The potential outcomes framework and Assumption~\ref{assumption:iv} are commonly used in the context of estimating complier treatment effect \citep{Abadie2002,Abadie2003,ORR15}.
Assumption (A1) assumes that the instrumental variable, $V$, mimics a random assignment, conditional on $\bX$. Assumption (A2) requires that $V$ affects the potential outcomes only through its effects on the treatment $D$. Assumption (A3) guarantees that  $D$ and $V$ are correlated, conditional on $\bX$, and that each subject can have $V=0$ or $V=1$ with a non-zero probability conditional on $\bX$. Assumption (A4) excludes the existence of {\it defiers}. For more details on the use of instrumental variable in the context of treatment effect estimation, we refer the reader to \cite{HW19}.

\section{The Proposed Method\label{Model}}
\subsection{Complier Expected Shortfall Treatment Effect and Model Assumptions} 
Recall from Section~\ref{POFA} that  $D$ is an indicator of a binary exposure and $Y_d$ is the potential outcome given $D = d$.  
Under the potential outcomes framework, the complier quantile and expected shortfall at level $\alpha \in (0,1)$ can be formally defined as 
\[
Q_{\alpha}(Y_d|\bX, D_1>D_0)=\inf\{y\in\mathbb{R}:\Pr(Y_d\le y|\bX,D_1>D_0)\ge \alpha\}
\]
and 
\begin{equation}
S_{\alpha}(Y_d|\bX,D_1>D_0)=\mathbb{E}\{Y_d|\bX,D_1>D_0,Y_d\le Q_{\alpha}(Y_d|\bX,D_1>D_0)\}=\frac{1}{\alpha}\int_{0}^{\alpha}Q_{u}(Y_d|\bX,D_1>D_0)du,
\label{SQeq0.1}
\end{equation}
respectively. 

One popular causal estimand is the complier quantile treatment effect, which is commonly used to estimate the local causal effect for compliers \citep{Abadie2002}. Specifically, the complier quantile treatment effect at  $\alpha$th quantile is defined as
\begin{equation}
\text{CQTE}(\alpha)=Q_{\alpha}(Y_1|\bX,D_1>D_0)-Q_{\alpha}(Y_0|\bX,D_1>D_0).
\label{SQeq0.2}
\end{equation}
The complier quantile treatment effect represents the difference between the compliers' $\alpha$th quantile of the potential outcomes $Y_1$ and that of $Y_0$. 
Parallel to~\eqref{SQeq0.2}, we define the complier expected shortfall treatment effect as 
\begin{equation}
\text{CRESTE}(\alpha)=S_{\alpha}(Y_1|\bX,D_1>D_0)-S_{\alpha}(Y_0|\bX,D_1>D_0),
\label{SQeq0.3}
\end{equation}
that is, the difference between compliers' $\alpha$th expected shortfall of the potential outcomes $Y_1$ and that of $Y_0$, given  covariates $\bX$. 
Under~\eqref{SQeq0.1}--\eqref{SQeq0.3}, CRESTE can be rewritten as
\begin{equation}
\text{CRESTE}(\alpha)=\frac{1}{\alpha}\int_{0}^{\alpha}\mathrm{CQTE}(u)du.
\label{SQeq0.4}
\end{equation}
In other words, CRESTE provides an average effect in the lower-tail of the distribution, rather than at a specific quantile level. 
Similarly, if the upper tail of the distribution is of interest we can  replace $S_{\alpha}(Y_d|\bX,D_1>D_0)$ with $S_{\alpha}^u(Y_d|\bX,D_1>D_0)=(1-\alpha)^{-1} \int_{\alpha}^{1}Q_{u}(Y_d|\bX,D_1>D_0)du$  in~\eqref{SQeq0.3}.

\begin{remark}
The quantity CRESTE($\alpha$) is equivalent to the expected shortfall treatment effect for the treated population ($D=1$) in the case of one-sided compliance, where subjects with $V=0$ have no access to treatment, i.e., $\Pr(D_0=0|\bX)=1$.  This is shown formally  in Proposition~\ref{SQprop1} in the Appendix. 
In practice, one-sided compliance occurs in many scientific studies. 
For example, in an observational study comparing a new drug versus placebo where the instrumental variable is chosen as whether treatment starts after the FDA approval date of the new drug, the one-sided compliance means that patients treated before the FDA approval of the new drug have no access to it.
\end{remark}

To estimate the complier expected shortfall treatment effect CRESTE($\alpha$), we assume the following models:
\begin{equation}
\begin{split}
Q_{\alpha}(Y_{d}|\bX,D_1>D_0)=&d\beta_{1}^*(\alpha)+\bX\bbeta^*_{\bX}(\alpha),\, ~d=0,1,\\
S_{\alpha}(Y_{d}|\bX,D_1>D_0)=&d\gamma_{1}^*(\alpha)+\bX\bgamma_{\bX}^*(\alpha),\,~d=0,1.
\end{split}
\label{SQeq1}
\end{equation}
Under \eqref{SQeq1}, $\beta_{1}^*(\alpha)$ and $\gamma_{1}^*(\alpha)$ can be interpreted as the complier $\alpha$th quantile and expected shortfall treatment effect, respectively. That is, 
\[ 
\begin{split}
\beta_{1}^*(\alpha)=&Q_{\alpha}(Y_{1}|\bX,D_1>D_0)-Q_{\alpha}(Y_{0}|\bX,D_1>D_0),\\
\gamma_{1}^*(\alpha)=&S_{\alpha}(Y_{1}|\bX,D_1>D_0)-S_{\alpha}(Y_{0}|\bX,D_1>D_0).
\end{split}
\]
Moreover, $\bbeta_{\bX}^*(\alpha)$ and $\bgamma_{\bX}^*(\alpha)$ quantify the effects of covariates $\bX$ on the conditional complier $\alpha$th quantile and expected shortfall of the potential outcome $Y_{d}$ given $\bX$, respectively.

Model~\eqref{SQeq1} involves the unobserved potential outcome $Y_d$. To provide a convenient venue to estimate $\beta_{1}^*(\alpha)$ and $\gamma_{1}^*(\alpha)$, we show in Proposition~\ref{SQprop2} in~\ref{SQpropappen} that  \eqref{SQeq1} is equivalent to
\begin{equation}
\begin{split}
Q_{\alpha}(Y|D,\bX,D_1>D_0)=&D\beta_{1}^*(\alpha)+\bX\bbeta^*_{\bX}(\alpha),\\
S_{\alpha}(Y|D,\bX,D_1>D_0)=&D\gamma_{1}^*(\alpha)+\bX\bgamma^*_{\bX}(\alpha),
\label{SQeq2}
\end{split}
\end{equation}
where $Y=D\times Y_1+(1-D)\times Y_0$, $Q_{\alpha}(Y|D,\bX,D_1>D_0)=\inf\{y:\Pr( Y\le y|D,\bX,D_1>D_0)\ge \alpha\}$, and $S_{\alpha}(Y|D,\bX,D_1>D_0)=\alpha^{-1}\int_0^\alpha Q_{u}(Y|D,\bX,D_1>D_0)du$. 
Under the reformulation in~\eqref{SQeq2}, the parameter of interests $\beta_{1}^*(\alpha)$ and $\gamma_{1}^*(\alpha)$ now depend only on the conditional quantile and expected shortfall of $Y$ rather than the potential outcome $Y_d$.

\begin{remark}
As will be elaborated in Section~\ref{SQasym}, the linearity assumption for the conditional quantile function in~\eqref{SQeq1} can be relaxed. For instance, one can replace the linear quantile model in~\eqref{SQeq1} with the  nonparametric quantile regression function with sufficient smoothness. The linearity assumption for conditional expected shortfall function in~\eqref{SQeq1} is adopted to balance model complexity and statistical interpretation.

\end{remark}

\subsection{Estimation Procedure\label{SQEPS}}

Let $\bZ=(D,\bX^\top)^\top$, $\bbeta=\{\beta_{1}(\alpha),\bbeta_{\bX}^\top(\alpha)\}^\top$, and $\bgamma=\{\gamma_{1}(\alpha),\bgamma_{\bX}^\top(\alpha)\}^\top$, where we suppress the dependency on $\alpha$ for notational convenience. Let $\{Y_1,V_1,\bZ_1^\top\}^\top,\ldots, \{Y_n,V_n,\bZ_n^\top\}^\top$ be $n$ independent and identically distributed realizations  of $\{Y,V,\bZ^\top\}^\top$. 
Recall that $\bbeta^*$ and $\bgamma^*$ are the true underlying values of $\bbeta$ and $\bgamma$, respectively.  Note that  $\bbeta^\ast=\text{argmin}_{\bbeta}\mathbb{E}\{I(D_1>D_0)\rho_{\alpha}(Y-\bZ^\top\bbeta)\}$, where $\rho_{\alpha}(u)=u\{\alpha-I(u\le 0)\}$ is the quantile loss function.
In addition, let
$g_{\alpha}(\bZ,\bb_1,\bb_2)=\bZ^\top\bb_1-\frac{1}{\alpha}(Y-\bZ^\top\bb_2)I(Y\le \bZ^\top\bb_2)-\bZ^\top\bb_2$ and $g_{\alpha,i}(\bb_1,\bb_2)=\bZ_i^\top\bb_1-\frac{1}{\alpha}(Y_i-\bZ_i^\top\bb_2)I(Y_i\le \bZ_i^\top\bb_2)-\bZ_i^\top\bb_2$.
To estimate $\bgamma^*$, one key observation  is that under~\eqref{SQeq2}, we have
\begin{equation}
\begin{split}
\mathbb{E}&\{I(D_1>D_0)\bZ g_{\alpha}(\bZ,\bgamma^\ast,\bbeta^\ast)\}=\bzero, \quad \mathrm{where}\quad \bbeta^\ast=\text{argmin}_{\bbeta}\mathbb{E}\{I(D_1>D_0)\rho_{\alpha}(Y-\bZ^\top\bbeta)\}.
\label{SQeq2.eq}
\end{split}
\end{equation}

However, $\bgamma^\ast$ and $\bbeta^\ast$ can not be directly estimated via~\eqref{SQeq2.eq} since $D_1$ and $D_0$ are not observed simultaneously.  Let $\kappa_v(Y,\bZ)=\Pr(D_1>D_0|Y,\bZ)$ be the conditional probability of the complier group, conditional on $Y$ and $\bZ$. 
By the law of iterated expectation, \eqref{SQeq2.eq} is equivalent to 
\begin{equation}
\begin{split}
\mathbb{E}&\{\kappa_v(Y,\bZ)\bZ g_{\alpha}(\bZ,\bgamma^\ast,\bbeta^\ast)\}=\bzero, \quad \mathrm{where} \quad \bbeta^\ast=\text{argmin}_{\bbeta}\mathbb{E}\{\kappa_v(Y,\bZ)\rho_{\alpha}(Y-\bZ^\top\bbeta)\}.
\label{SQEP0}
\end{split}
\end{equation}
Equation~\eqref{SQEP0} suggests a simple weighting scheme for estimating $\bbeta^\ast$ and $\bgamma^\ast$. 
Let $D_{i1}$ and $D_{i0}$ be the potential treatment selection for the $i$th sample given $V_i=1$ and $V_i=0$,  respectively. Let $\kappa_{v,i}= \Pr(D_{i1}>D_{i0}|Y_i,\bZ_i)$ be the conditional probability of the $i$th sample.  Then, an estimating equation for $\bgamma^\ast$ is given by 
\begin{equation}
\sum_{i=1}^n\kappa_{v,i}\bZ_i g_{\alpha,i}(\bgamma,\hat{\bbeta})=\bzero, \quad \mathrm{where}\quad  \hat\bbeta=\text{argmin}_{\bbeta}\sum_{i=1}^n\kappa_{v,i}\rho_{\alpha}(Y_i-\bZ_i^\top\bbeta). 
\label{SQEP01}
\end{equation}

In practice, $\kappa_{v,i}$ is unknown and needs to be estimated.  
By Proposition~\ref{SQprop3} in the Appendix, we have
\begin{equation}
\kappa_v(Y,\bZ)=1-\frac{D\{1-v(Y,\bZ)\}}{1-\pi(\bX)}-\frac{(1-D)v(Y,\bZ)}{\pi(\bX)},
\label{SQweight1}
\end{equation}
where $v(Y,\bZ)=\Pr(V=1|Y,\bZ)$ and $\pi(\bX)=\Pr(V=1|\bX)$.  Note that $\kappa_v(Y,\bZ)$ is identical to  $\kappa_v(\bU)$ in Lemma 3.2 of \cite{Abadie2002}. 
Thus, it suffices to estimate $\pi(\bX)$ and  $v(Y,\bZ)$.  
One widely used approach to model the conditional distribution of $V$ is the logistic regression, which may suffer from model misspecification.  
We instead use an alternative non-parametric approach for modeling $V$ given  $Y$, $D$, and $\bX$ \citep{Wei2021}.  


Briefly, let $v_{d}(Y,\bX)=\Pr(V=1|Y,\bX,D=d)$ and let $v(Y_i,\bZ_i)=
I( D_i=0)v_{0}(Y_i,\bX_i)+I( D_i=1)v_{1}(Y_i,\bX_i)$. Denote $\mathcal{K}^{\ast}_{\sigma_1}(\bu)$ and $\mathcal{K}^{\ast\ast}_{\sigma_2}(\bu)$ as two kernel functions that satisfy Conditions (C7)--(C8) in Section~\ref{SQasym} with bandwidths $\sigma_1$ and $\sigma_2$, respectively. 
When all components of $\bX$ are continuous, we propose to estimate $\pi(\bx)$ and $v_{d}(y,\bx)$ via
\begin{equation}
\label{eq:discreteXvscontinousX}
\hat\pi(\bx)=\frac{\sum_{i=1}^n\mathcal{K}^{\ast}_{\sigma_1}(\bx-\bX_i)V_i}{\sum_{i=1}^n\mathcal{K}^{\ast}_{\sigma_1}(\bx-\bX_i)}\ \ {\rm and}\ \ 
\hat v_{d}(y,\bx)=\frac{\sum_{i=1}^n I(D_i=d)\mathcal{K}^{\ast\ast}_{\sigma_2}\{(y,\bx^\top)^\top-(Y_i,\bX_i^\top)^\top\}V_i}{\sum_{i=1}^n I(D_i=d)\mathcal{K}^{\ast\ast}_{\sigma_2}\{(y,\bx^\top)^\top-(Y_i,\bX_i^\top)^\top\}},
\end{equation}
respectively.
Subsequently, we estimate $v(Y_i,\bZ_i)$ via $\hat v(Y_i,\bZ_i)=I(D_i=1)\hat v_{1}(Y_i,\bX_i)+I(D_i=0)\hat v_{0}(Y_i,\bX_i)$. 
Thus, a non-parametric estimator of $\kappa_{v,i}$ is then given by
\begin{equation}
\label{eq:kappatilde}
\hat\kappa_{v,i}=1-\frac{D_i\{1-\hat v(Y_i,\bZ_i)\}}{1-\hat \pi(\bX_i)}-\frac{(1-D_i)\hat v(Y_i,\bZ_i)}{\hat \pi(\bX_i)}.
\end{equation}

Note that $\kappa_{v,i}$ is a value between zero and one since it is the conditional probability of the complier group for the $i$th sample.  
We enforce such constraints on $\hat\kappa_{v,i}$ by performing truncation around zero and one, i.e., 
\begin{equation}
\label{eq:kappatilde2}  
\tilde\kappa_{v,i}=\min\{\max(\hat\kappa_{v,i}, c_{l,n}), c_{u, n}\},
\end{equation}
where $c_{l,n}$ and $c_{u,n}$ are two sequences of positive constants that approach to zero and one, respectively, as the sample size $n$ increases. 
Replacing $\kappa_{v,i}$ with the non-parametric estimator $\tilde{\kappa}_{v,i}$ in~\eqref{SQEP01}, an estimator for $\bgamma^\ast$ can then be obtained by solving the estimating equation
\begin{equation}
\sum_{i=1}^n\tilde{\kappa}_{v,i}\bZ_i g_{\alpha,i}(\bgamma,\hat{\bbeta})=\mathbf{0}, \quad \mathrm{where}\quad  \hat\bbeta=\text{argmin}_{\bbeta}\sum_{i=1}^n\tilde{\kappa}_{v,i}\rho_{\alpha}(Y_i-\bZ_i^\top\bbeta),
\label{SQEP3}
\end{equation}
which amounts to solving a weighted ordinary least squares problem. 

The proposed estimation procedure for estimating the complier expected shortfall treatment effect is summarized as follows:
\begin{itemize}
\setlength{\itemindent}{0.5in}
\item[Step 1:] Calculate $\hat{\pi}(\bX_i)$ and $\hat{v}(Y_i,\bZ_i)$ with bandwidths $\sigma_1$ and $\sigma_2$ selected via cross-validation, and obtain $\tilde{\kappa}_{v,i}$ for $ i=1,\ldots,n$.
\item[Step 2:] Calculate the quantile regression estimator $\hat\bbeta=\text{argmin}_{\bbeta}\sum_{i=1}^n\tilde{\kappa}_{v,i}\rho_{\alpha}(Y_i-\bZ_i^\top\bbeta)$.
\item[Step 3:] Plug $\tilde\kappa_{v,i}$ and $\bZ_i^\top\hat\bbeta$ into~\eqref{SQEP3}, and obtain $\hat\bgamma$ by solving~\eqref{SQEP3}.
\end{itemize}
Details for performing statistical inference on $\bgamma^*$ are deferred to Section~\ref{SQasym}. 

\begin{remark}
The series-based non-parametric estimator used in \cite{Abadie2002} is another non-paramteric approach for estimating  $\kappa_{v,i}$. By using  similar regularity conditions and techniques to those  in \cite{Abadie2002}, we can show that the asymptotic results in Section 4 remain valid with the series-based estimator of   $\kappa_{v,i}$.  Generally, which nonparametric estimator one uses to estimate $\kappa_{v,i}$ does not have a strong impact on the estimation of CRESTE as long as the tuning parameters are chosen appropriately and the dimension of $\bX$ is moderate.
\end{remark}

\begin{remark}
\label{remark:bandwidth}
In practice,  the bandwidth parameters $\sigma_{1}$ and $\sigma_{2}$ can be selected using cross-validation. 
Let $\hat{\pi}_{\sigma}(\bX)$ be an estimator of $\pi(\bX)$ obtained from the training dataset with
bandwidth $\sigma$. 
We can then select $\sigma_{1,n }$ as the value of $\sigma$ that minimizes $\sum_{i\in\text{test set}}|V_i-\hat\pi_{\sigma}(\bX_i)|$. The bandwidth $\sigma_{2,n}$ can be selected similarly.
\end{remark}

\begin{remark}

For discrete covariates in $\bX$,  $\pi(\bX_i)$ and $v(Y_i,\bZ_i)$ can be estimated by first stratifying the dataset into multiple cells based on the discrete covariates in $\bX$.  For each cell,  non-parametric estimates of $\pi(\bX_i)$ and $v(Y_i,\bZ_i)$ can be constructed by using only the continuous covariates in $\bX$ and $Y$ using~\eqref{eq:discreteXvscontinousX}, since the discrete covariates take the same value within each cell.
When all of the covariates $\bX$ are discrete, non-parametric estimates for $\pi(\bX_i)$ in each cell reduce to empricial estimates of $E(V)$ in each cell.  Non-parametric estimates for $v(Y_i,\bZ_i)$ in each cell can be obtained by only using the continuous variables in $Y$ in its corresponding cell.

\end{remark}

\section{Asymptotic Properties}\label{SQasym}

We now establish consistency and weak convergence of the proposed estimators $\hat\bbeta$ and $\hat\bgamma$ at the $\alpha$th level quantile and expected shortfall, for any given $\alpha\in (0,1)$. Let $f(\cdot|\bZ,D_1>D_0)$ be the density function of $Y$ conditional on $\bZ$ and $D_1>D_0$. For notational convenience, we write $v =  v(Y,\bZ)$, $v_i=  v(Y_i,\bZ_i)$, $\pi=  \pi(\bX)$, and $\pi_i= \pi(\bX_i)$. Besides, let $\bX_c$ be a subvector of continuous covariates in $\bX$.  We start with some regularity conditions. 

\begin{itemize}
\setlength{\itemindent}{-0in}
\item[(C1)] The data $(Y_i,D_i,\bX_i,V_i)$ are independent and identically distributed for $i=1,2,\ldots, n$. 

\item[(C2)] (a) Each discrete component of $\bX$ takes on finitely many values; (b) conditional on $D$, $D_1>D_0$, and the discrete components of $\bX$,  $(Y,\bX_c)$ has a support as a closed subset of the product of compact intervals and has a density at least third order continuously differentiable and bounded away from zero and infinity; (c) $\bbeta^\ast\in\mathcal{K}$ and $\bgamma^\ast\in\mathcal{B}$, where $\mathcal{K}$ and $\mathcal{B}$ are compact subsets in $\mathbb{R}^{l+1}$.

\item[(C3)] $\mathbb{E}\{I(D_1>D_0)\bZ\bZ^\top\}$ is of full rank.

\item[(C4)] For all $\bZ$, there exists some $c_0>0$, such that $f(\bZ^\top\bbeta^\ast|\bZ,D_1>D_0)>c_0$.

\item[(C5)] (a) For some $c>0$, $\kappa_v(Y,\bZ)>c$ almost surely; (b) For some $0<c_1<c_2<1$, $c_1<\pi(\bX)<c_2$ almost surely.

\item[(C6)] The sequences $c_{l,n}>0$ and $c_{u,n}<1$ satisfy $c_{l,n}=o(n^{-1/2})$ and $1-c_{u,n}=o(n^{-1/2})$.

\item[(C7)] (a) There is a positive integer $\Delta$, such that $\mathcal{K}_{\sigma_1}^{\ast}(\bu)$ and $\mathcal{K}^{\ast\ast}_{\sigma_2}(\bu)$ are differentiable of order $\Delta$ and the derivatives of order $\Delta$ are Lipschitz in a bounded support.  $\mathcal{K}_{\sigma_1}^{\ast}(\bu)$ and $\mathcal{K}^{\ast\ast}_{\sigma_2}(\bu)$  have bounded support; (b) $\mathcal{K}_{\sigma_k}^{\ast}(\bu)=1$  for $k=1,2$; (c) for some positive integers $s_1$ and $s_2$,  $\int\mathcal{K}^\ast_{\sigma_k}(\bu)[\otimes_{r=1}^j \bu]d\bu=0$ for all $j<s_k$, where $k=1,2$, and $\otimes_{r=1}^j \bu$ stands for executing $j$ times Kronecker product on $\bu$.

\item[(C8)] There exists a constant $p$ such that:  (a) $v(\cdot)$ and $\pi(\cdot)$ are at least $p$th order continuous differentiable;   (b) $p\ge s_k, k=1,2$; and (c) $n\sigma_{i,n}^{2p}\to 0$ and  $\frac{n\sigma_{i,n}^{2l+2}}{(\log n)^2}\to\infty$. 
\end{itemize}

Condition (C2) implies the boundedness of $\bX$ and $\bZ$, and the positiveness
and boundedness of the density of $Y$ or $\bX$. Conditions (C3) and (C4) are imposed to ensure the identifiablilty of $\bbeta^\ast$ and $\bgamma^\ast$; similar conditions have been used in \cite{Barendse2020}. By Condition (C5), $\kappa_v(Y,\bZ)$ and $\pi(\bX)$
are bounded away from zero and one almost surely. Condition (C6) implies that truncating $\hat{\kappa}_v$ by $c_{l,n}$ and $c_{u,n}$ would only lead to a negligible impact on the asymptotic results
with $c_{l,n}$ and $c_{u,n}$ approaching zero and one, respectively. 
Conditions (C7)--(C8)  are similar to the regularity conditions in \cite{Newey1994} for kernel estimators. By Condition (C8), we require that $v(\cdot)$ and $\pi(\cdot)$ to be smooth, and that the bandwidths satisfy $\sigma_{k,n}=o(n^{-1/(2p)}\wedge(\log n)^{1/(l+1)}n^{-1/(2l+2)})$ for $k=1,2$. Compared to existing regularity conditions in the context of estimating CQTE \citep{Abadie2002}, the proposed CRESTE estimator only additionally  requires the full rank condition in~(C4) to ensure  the identifiablilty of $\bgamma^\ast$. In practice, Condition (C3) is satisfied in many studies.

Given the regularity conditions, we now establish the theoretical properties of the proposed estimators $\hat{\bbeta}$ and $\hat{\bgamma}$ in Theorems~\ref{SQthm0}--\ref{SQthm2}. Let $\bfm_1(Y,\bZ,\alpha)=\bZ\{\alpha-I(Y<\bZ^\top\bbeta^\ast)\}$ and let  $\bH_1(\bX,\alpha)=\mathbb{E}\left[\bfm_1(Y,\bZ,\alpha)\left\{\frac{(1-D)v}{\pi^2}-\frac{D(1-v)}{(1-\pi)^2}\right\}|\bX\right]$.  In addition, let
\[ 
\bPhi(\alpha)=\bfm_1(Y,\bZ,\alpha)\left\{1-\frac{D(1-V)}{1-\pi}-\frac{(1-D)V}{\pi}\right\}+\bH_1(\bX,\alpha)\{V-\pi(\bX)\}.
\]
The following theorem establishes the asymptotic normality of $\hat{\bbeta}$.

\begin{theorem}[Consistency and asymptotic distribution of $\hat{\bbeta}$]
Under Conditions (C1)--(C8), we have 
\[
n^{1/2}(\hat\bbeta-\bbeta^\ast)\to_d N(\bzero, \bJ_1^{-1}\bOmega_1\bJ_1^{-1}),
\]
where $\bOmega_1=\mathbb{E}\{\bPhi(\alpha)\bPhi(\alpha)^\top\}$  and $\bJ_1=\mathbb{E}\{I(D_1>D_0)f(\bZ^\top\bbeta^\ast|\bZ,D_1>D_0)\bZ\bZ^\top\}$. 
\label{SQthm0}
\end{theorem}

Next, we establish consistency and asymptotic normality for $\hat{\bgamma}$ obtained from solving~\eqref{SQEP3}.   Let $\bfm_2(Y,\bZ,\alpha)=\bZ g_{\alpha}(\bZ,\bgamma^\ast,\bbeta^\ast)$ and let $\bH_2(\bX,\alpha)=\mathbb{E}\left[\bfm_2(Y,\bZ,\alpha)\left\{\frac{(1-D)v}{\pi^2}-\frac{D(1-v)}{(1-\pi)^2}\right\}|\bX\right]$.  In addition, let $\bPsi(\alpha)=\bfm_2(Y,\bZ,\alpha)\left\{1-\frac{D(1-V)}{1-\pi}-\frac{(1-D)V}{\pi}\right\}+\bH_2(\bX,\alpha)\{V-\pi(\bX)\}$.
Theorems~\ref{SQthm1} and~\ref{SQthm2} establish the consistency and asymptotic normality of $\hat{\bgamma}$, respectively.

\begin{theorem}[Consistency of $\hat{\bgamma}$]
Under Conditions (C1)--(C8), we have 
\[
\|\hat\bgamma-\bgamma^\ast\|\to_p 0,
\]
where $\|\cdot\|$ is the $\ell_2$ norm. 
\label{SQthm1}
\end{theorem}

\begin{theorem}[Asymptotic distribution of $\hat{\bgamma}$]
Under Conditions (C1)--(C3) and (C5)--(C8), if $\|\hat\bbeta-\bbeta^\ast\|=o_p(n^{-1/4})$,  we have 
\[
n^{1/2}(\hat\bgamma-\bgamma^\ast)\to_d N(\bzero, \bJ_2^{-1}\bOmega_2\bJ_2^{-1}),
\]
where  $\bOmega_2=\mathbb{E}\{\bPsi(\alpha)\bPsi(\alpha)^\top\}$ and $\bJ_2=\mathbb{E}\{I(D_1>D_0)\bZ\bZ^\top\}$.
\label{SQthm2}
\end{theorem}

Theorem~\ref{SQthm2} requires only  $\|\hat\bbeta-\bbeta^\ast\|=o_p(n^{-1/4})$ to establish asymptotic normality  of $\hat{\bgamma}$: this is in contrast to the faster rate  $\|\hat\bbeta-\bbeta^\ast\|=O_p(n^{-1/2})$  needed in \cite{Chen2021} to obtain similar results. 
The weaker convergence requirement for $\hat{\bbeta}$ allows 
us to replace the linear quantile model in~\eqref{SQeq1} with other flexible approaches, such as the nonparametric quantile regression method \citep{He1994,De2003}, without affecting the asymptotic distribution of  $\hat{\bgamma}$. 
In addition, \cite{Chen2021} required $O_p(n^{-1/2})$ as the convergence rate for $\hat\pi(\bX)$ to establish the asymptotic distribution of $\hat{\bgamma}$. Typically, the above convergence rate can only be satisfied when a parametric estimator is used to estimate $\pi(\bX)$.  We instead adopt a non-parametric estimator of $\pi(\bX)$ with convergence rate $o_p(n^{-1/4})$ to establish the asymptotic distribution of $\hat\bgamma$. The non-parametric estimate of $\pi(\bX)$ does not require parametric modeling, and thus prevents the bias induced from model misspecifications.

\begin{remark}
\label{statinference}
From Theorem \ref{SQthm2}, we have  $n^{1/2}(\hat\bgamma-\bgamma^\ast)\to_d N(\bzero, \bJ_2^{-1}\bOmega_2\bJ_2^{-1})$. In principle, the sample-based variance estimator can be used to estimate $\bJ_2$ and $\bOmega_2$. However, $\bOmega_2$ includes evaluating the conditional expectation given $\bX$, $\bH_2(\bX,\alpha)$; this quantity is challenging to estimate when $\bX$ is continuous. 
We instead propose to perform statistical inference on $\bgamma^\ast$ using the bootstrap method. 
In particular, to estimate the asymptotic variance of  $\hat\bgamma$,  we resample $n$ data points with replacement, and obtain an estimator for $\bgamma^\ast$ by applying the proposed estimation method described in Section~\ref{SQEPS} to the resampled dataset. Let $\{\hat\bgamma_b\}_{b=1}^B$ be the estimates obtained by repeating the aforementioned bootstrap method for $B$ number of times. The variance of $\hat\bgamma$ can then be approximated by the empirical variance of $\{\hat\bgamma^\ast_b\}_{b=1}^B$. A formal justification for the presented nonparametric bootstrap inference procedure is provided in Appendix C.
\end{remark}

\section{Numerical Studies} \label{simu}
We assess the finite-sample performance of the proposed method via extensive numerical studies. 
We compare the proposed method to several approaches: (i) the naive method by directly applying the two-stage method to the entire dataset \citep{Barendse2020}; (ii) the oracle (but practically infeasible) method that applies the two-stage method to the latent complier subgroup that is unknown in practice; and (iii) the  joint regression method for estimating complier expected shortfall effect by \citet{Chen2021}.
Our proposed method involves estimating the conditional probability of complier group conditional on $Y_i$ and $\bZ_i$, i.e., $\kappa_{v,i}$.
We compute the non-parametric estimator $\tilde\kappa_{v,i}$ \eqref{eq:kappatilde2} using the second order Epanechnikov kernel with $c_{l, n}=10/n$ and $c_{u, n}=1-10/n$. The bandwidths for estimating $\hat\pi(\bX_i)$ and $\hat v(Y_i,\bZ_i)$ are selected from $\{0.1,0.2,\ldots,0.9\}$ based on cross-validation with the criterion described in Remark~\ref{remark:bandwidth}. 
We use the bootstrap approach to estimate the standard errors for all of the estimators with $B=1000$ number of bootstrapped samples. 
To assess the performance across different methods, we compute the average bias of the estimated coefficients, average variance estimates and the empirical variances, and coverage probabilities of the $95\%$ confidence interval estimates. 


We first  generate 
the latent compliance group memberships from a multinomial distribution such that $\Pr(\mathrm{Compliers})=2/3$ and $\Pr(\mathrm{Always\ takers})=\Pr(\mathrm{Never\ takers})=1/6$. 
We then generate two independent covariates $X_1\sim \mathrm{Unif}(0,1)$ and  $X_2\sim \mathrm{Bernoulli}(0.5)$. 
Given $\bX=(X_1,X_2)^{\top}$, the instrumental variable $V$ is generated from a Bernoulli distribution with probability
\[
\pi(\bX,\epsilon)=\frac{\exp(0.1X_2+X_1^2+X_1X_2+\epsilon)}{1+\exp(0.1X_2+X_1^2+X_1X_2+\epsilon)},
\]
where $\epsilon\sim N(0,0.5^2)$ controls the deviation of $\pi(\bX)$ from a logistic regression model. The treatment variable $D$ can then be determined based on $V$ and the latent compliance subgroup membership via the following equation
\[
D=\left\{
\begin{array}{cl}
V,&\, \mathrm{Compliers},\\
1,&\, \mathrm{Always\ takers},\\
0,&\, \mathrm{Never\ takers}.
\end{array}
\right.
\]

With $D$, $V$, and the latent compliance memberships, the response  $Y$ for compliers and non-compliers are generated from the following model:
$$
Y=\left\{\begin{array}{ll}
\log\tau-0.2X_1-0.3X_2+0.5\times\exp(0.3\tau)\times D, & \text{Compliers},\\
-0.1X_1-0.2X_2+0.2D+\epsilon_{\mathrm{nc}},&  \text{Otherwise},\\
\end{array}\right. 
$$
where $\tau\sim \mathrm{Unif}(0,1)$ and $\epsilon_{\mathrm{nc}}\sim N(0,0.5^2).$
We note that the data generating mechanism satisfies Assumptions (A1)--(A4).

Results with $n=\{ 500,3000\}$ for $\alpha = \{0.1,0.2,0.3,0.4,0.5\}$, averaged across 1000 replications, are presented in Table~\ref{SQsimu:t1}. 
We see that the performance of the proposed method is close to that of the oracle method, i.e, the estimated parameters of interest are close to their corresponding true underlying values.
Moreover, the empirical coverage probabilities of 95\% confidence intervals are close to the nominal level, and the bootstrap-based variance estimates agree well with the empirical variances.
  In contrast, the naive method produces substantially biased estimators and suffers from under-coverage. The numerical results confirm that ignoring treatment endogeneity can lead to biased estimation and inference. We see that as we increase the sample size to $n=3000$, the bias of the proposed method further diminishes. 

\begin{table}[!t]
	\fontsize{8}{9}\selectfont
	\centering
	\caption{Comparisons among estimators of $\beta_{1}(\alpha)$ and $\gamma_{1}(\alpha)$ from the proposed two-stage method, oracle method and naive methods in the simulation experiment with $n \in \{500,3000\}$ and $\alpha=\{0.1,0.2,\ldots, 0.5\}$. Bias, Emp var, Boot var and Cov 95 stand for average bias of the estimated coefficients, empirical variance, average variance estimates, coverage probabilities of the 95\% confidence intervals.  }
	\label{SQsimu:t1}
	 \begin{threeparttable}
	
\begin{tabular}{|c|c|l|cc|cc|cc|}
    \hline
\multirow{2}{*}{$\alpha$} &\multirow{2}{*}{$n$}&&\multicolumn{2}{c|}{Oracle method}  & \multicolumn{2}{c|}{Proposed method} &\multicolumn{2}{c|}{Naive method} \\
 \cline{3-9}
&  & &$\beta_{1}(\alpha)$&$\gamma_{1}(\alpha)$&$\beta_{1}(\alpha)$&$\gamma_{1}(\alpha)$&$\beta_{1}(\alpha)$&$\gamma_{1}(\alpha)$\\
\hline
 \multirow{8}{*}{0.1}&\multirow{4}{*}{500}&Bias&-0.006&0.023&-0.036&-0.047&-0.197&-0.208\\
     \cline{3-9}
     && Emp var&0.123&0.262&0.124&0.247&0.080&0.161\\   
      \cline{3-9}
     && Boot var&0.143&0.244&0.140&0.230&0.088&0.156\\   
       \cline{3-9}
      && Cov 95&0.955&0.931&0.948&0.931&0.892&0.896\\ 
       \cline{2-9}
       &{\multirow{4}{*}{3000}}&Bias&0.000&-0.003&-0.015&-0.018&-0.195&-0.200\\
     \cline{3-9}&& Emp var&0.022&0.047&0.021&0.043&0.013&0.029\\   
      \cline{3-9}&& Boot var&0.022&0.043&0.022&0.042&0.014&0.027\\   
      \cline{3-9}&& Cov 95&0.938&0.945&0.946&0.931&0.624&0.756\\  
      \hline
      \multirow{8}{*}{0.2}&\multirow{4}{*}{500}&Bias&0.001&-0.022&-0.039&-0.069&-0.201&-0.210\\
     \cline{3-9}&& Emp var&0.052&0.118&0.056&0.117&0.029&0.071\\   
      \cline{3-9}&& Boot var&0.062&0.120&0.063&0.116&0.033&0.073\\   
      \cline{3-9}&& Cov 95&0.961&0.943&0.946&0.935&0.789&0.874\\
      \cline{2-9}
       &{\multirow{4}{*}{3000}}&Bias&0.002&-0.002&-0.011&-0.014&-0.196&-0.197\\
     \cline{3-9}&& Emp var&0.009&0.019&0.009&0.021&0.005&0.012\\   
     \cline{3-9}&& Boot var&0.010&0.021&0.010&0.020&0.005&0.013\\   
     \cline{3-9}&& Cov 95&0.954&0.954&0.947&0.939&0.242&0.574\\   
      \hline
  \multirow{8}{*}{0.3}&\multirow{4}{*}{500}&Bias&0.001&0.001&-0.011&-0.023&-0.207&-0.202\\
     \cline{3-9}&& Emp var&0.033&0.080&0.034&0.080&0.014&0.046\\   
     \cline{3-9}&& Boot var&0.038&0.079&0.039&0.077&0.016&0.046\\   
      \cline{3-9}&& Cov 95&0.955&0.943&0.968&0.943&0.625&0.824\\  
      \cline{2-9}
       &{\multirow{4}{*}{3000}}&Bias&0.001&0.002&-0.015&-0.014&-0.202&-0.198\\
     \cline{3-9}&& Emp var&0.005&0.012&0.006&0.014&0.002&0.008\\   
     \cline{3-9}&& Boot var&0.006&0.013&0.006&0.013&0.002&0.008\\   
      \cline{3-9}&& Cov 95&0.965&0.965&0.938&0.937&0.026&0.376\\     
      \hline      
      
  \multirow{8}{*}{0.4}&\multirow{4}{*}{500}&Bias&-0.005&-0.008&-0.021&-0.028&-0.211&-0.207\\
     \cline{3-9}&& Emp var&0.023&0.055&0.022&0.053&0.009&0.032\\   
      \cline{3-9}&& Boot var&0.025&0.055&0.025&0.054&0.009&0.031\\   
      \cline{3-9}&& Cov 95&0.944&0.938&0.955&0.948&0.396&0.766\\  
      \cline{2-9}
       &{\multirow{4}{*}{3000}}&Bias&-0.007&-0.008&-0.012&-0.007&-0.208&-0.207\\
     \cline{3-9}&& Emp var&0.003&0.009&0.004&0.009&0.001&0.005\\   
      \cline{3-9}&& Boot var&0.004&0.009&0.004&0.009&0.002&0.005\\   
      \cline{3-9}&& Cov 95&0.958&0.960&0.951&0.951&0.000&0.160\\       
      \hline

      \multirow{8}{*}{0.5}&\multirow{4}{*}{500}&Bias&0.007&0.005&-0.034&-0.031&-0.204&-0.195\\
     \cline{3-9}&& Emp var&0.015&0.045&0.014&0.039&0.005&0.024\\   
      \cline{3-9}&& Boot var&0.017&0.041&0.017&0.041&0.006&0.022\\   
      \cline{3-9}&& Cov 95&0.952&0.935&0.952&0.953&0.263&0.723\\  
      \cline{2-9}
       &{\multirow{4}{*}{3000}}&Bias&0.001&-0.004&-0.018&-0.017&-0.209&-0.206\\
     \cline{3-9}&& Emp var&0.003&0.007&0.003&0.007&0.001&0.004\\   
     \cline{3-9}&& Boot var&0.003&0.007&0.003&0.007&0.001&0.004\\   
      \cline{3-9}&& Cov 95&0.949&0.943&0.937&0.940&0.000&0.077\\    
      \hline      
\end{tabular}
     \end{threeparttable}
\end{table}

Next, we compare the numerical performance between the proposed two-stage method and the joint regression method in \cite{Chen2021} with different choices of specification functions $G_1(\cdot)$ and $G_2(\cdot)$. Specifically, we consider  four combinations of $(G_1(\cdot),G_2(\cdot))$, where $G_1(z)=z$ or $0$, and $G_2(z)=\exp(z)$ or $\log\{1+\exp(z)\}$. The initial values for the joint regression method are generated from the normal distribution with mean equal to the estimator obtained from the proposed two-stage method and standard deviation equal to one. The non-parametric estimated conditional probability of complier group described above, $\tilde{\kappa}_{v}$, are adopted in both methods.  In each setting, we generate 1000 simulated data sets and choose $B=1000$ as the number of bootstrapped samples. 

Results with $n=3000$ for $\alpha = \{0.1,0.2,0.3,0.4,0.5\}$ are presented in Table~\ref{SQsimu:t2}. 	Histograms of $\hat{\gamma}_{1}(\alpha)$ across the 1000 simulated data sets are also shown in Figure~\ref{SQsimu:f1}. 
We found from Table~\ref{SQsimu:t2} and Figure~\ref{SQsimu:f1}, that the proposed two-stage method and joint regression method have two main differences in their numerical performance: (1) the distributions of 1000 estimators from the proposed two-stage method across all $\alpha$ levels are well approximated to a normal distribution, while the  distributions of 1000 estimators from the joint regression method are not  well approximated by a normal distribution for some $\alpha$ levels;
(2) the estimated coefficients from the proposed two-stage method are close to the true underlying value, while  the estimates from the joint regression method are not close to the true underlying values for some $\alpha$ levels.
One potential reason for the differences is that the joint regression estimators are obtained from solving a non-convex loss function for which global minimum is not guaranteed. 
Thus, the estimators may be sensitive to the choice of initial values.

In short, our numerical results suggest that without taking into account compliers can lead to substantial biased estimation and inference. Compared to the joint regression approach, our proposed two-stage method has a more robust and stable numerical performance.
Moreover, the proposed method can be implemented efficiently compared to that of the joint regression approach in \citet{Chen2021}.

\begin{table}[!t]
	\fontsize{8}{9}\selectfont
	\centering
	\caption{Results for the two stage method and the  joint regression method under four different combinations of $(G_1(\cdot),G_2(\cdot))$, denoted as Joint 1--4,  with $n=3000$ and $\alpha=\{0.1,0.2,\ldots, 0.5\}$. Bias, Emp var, Boot var, Boot SD, and Cov 95 stand for average bias of the estimated coefficients, empirical variance, average variance estimates, average standard error estimates, and coverage probabilities of the 95\% confidence intervals.}
	\label{SQsimu:t2}
	 \begin{threeparttable}
	
\begin{tabular}{|c|c|c|c|c|c|c|}
    \hline
 \multirow{2}{*}{$\alpha$}&&\multirow{2}{*}{Two-stage method} &\multicolumn{4}{c|}{Joint regression method} \\
 \cline{4-7}&&&Joint 1&Joint 2&Joint 3&Joint 4\\
 \hline
 \multirow{5}{*}{0.1}& Bias&-0.018&-0.004&-0.025&-0.029&-0.008\\
     \cline{2-7} &Emp var&0.043&0.046&0.199&0.198&0.045\\
      \cline{2-7} &Boot var&0.042&0.047&0.048&0.046&0.045\\
      \cline{2-7} &Cov 95&0.931&0.940&0.936&0.932&0.932\\
	   \cline{2-7}&Boot SD&0.204&0.215&0.217&0.213&0.211\\
	   \hline
	   
	    \multirow{5}{*}{0.2}&  Bias&-0.014&-0.036&-0.098&-0.099&-0.038\\
     \cline{2-7} &Emp var&0.021&0.157&0.958&0.961&0.159\\
      \cline{2-7}& Boot var&0.020&0.061&0.030&0.030&0.061\\
      \cline{2-7}& Cov 95&0.939&0.942&0.932&0.932&0.940\\
      \cline{2-7}& Boot SD&0.142&0.162&0.154&0.153&0.161\\
	   \hline
	   
	    \multirow{5}{*}{0.3}&   Bias&-0.014&-0.010&-0.120&-0.120&-0.011\\
    \cline{2-7}& Emp var&0.014&0.014&1.374&1.374&0.014\\
      \cline{2-7}& Boot var&0.013&0.020&0.022&0.022&0.020\\
     \cline{2-7}& Cov 95&0.937&0.939&0.929&0.928&0.938\\
      \cline{2-7}& Boot SD&0.114&0.120&0.125&0.123&0.120\\
	   \hline
	   
	    \multirow{5}{*}{0.4}&  Bias&-0.007&-0.031&-0.178&-0.178&-0.032\\
     \cline{2-7}& Emp var&0.009&0.168&3.196&3.196&0.169\\
      \cline{2-7}& Boot var&0.009&0.158&0.016&0.016&0.158\\
      \cline{2-7}& Cov 95&0.951&0.949&0.936&0.935&0.949\\
      \cline{2-7}& Boot SD&0.097&0.126&0.107&0.107&0.125\\
	   \hline
	   
	    \multirow{5}{*}{0.5}&  Bias&-0.017&-0.091&-0.342&-0.343&-0.092\\
     \cline{2-7}& Emp var&0.007&0.415&4.466&4.466&0.418\\
      \cline{2-7}& Boot var&0.007&0.290&0.019&0.019&0.290\\
      \cline{2-7}& Cov 95&0.940&0.937&0.910&0.910&0.935\\
      \cline{2-7}& Boot SD&0.085&0.151&0.101&0.100&0.151\\  
	   \hline
\end{tabular}
     \end{threeparttable}
\end{table}

\begin{figure}[!htp]
\centering
\includegraphics[scale=0.4543]{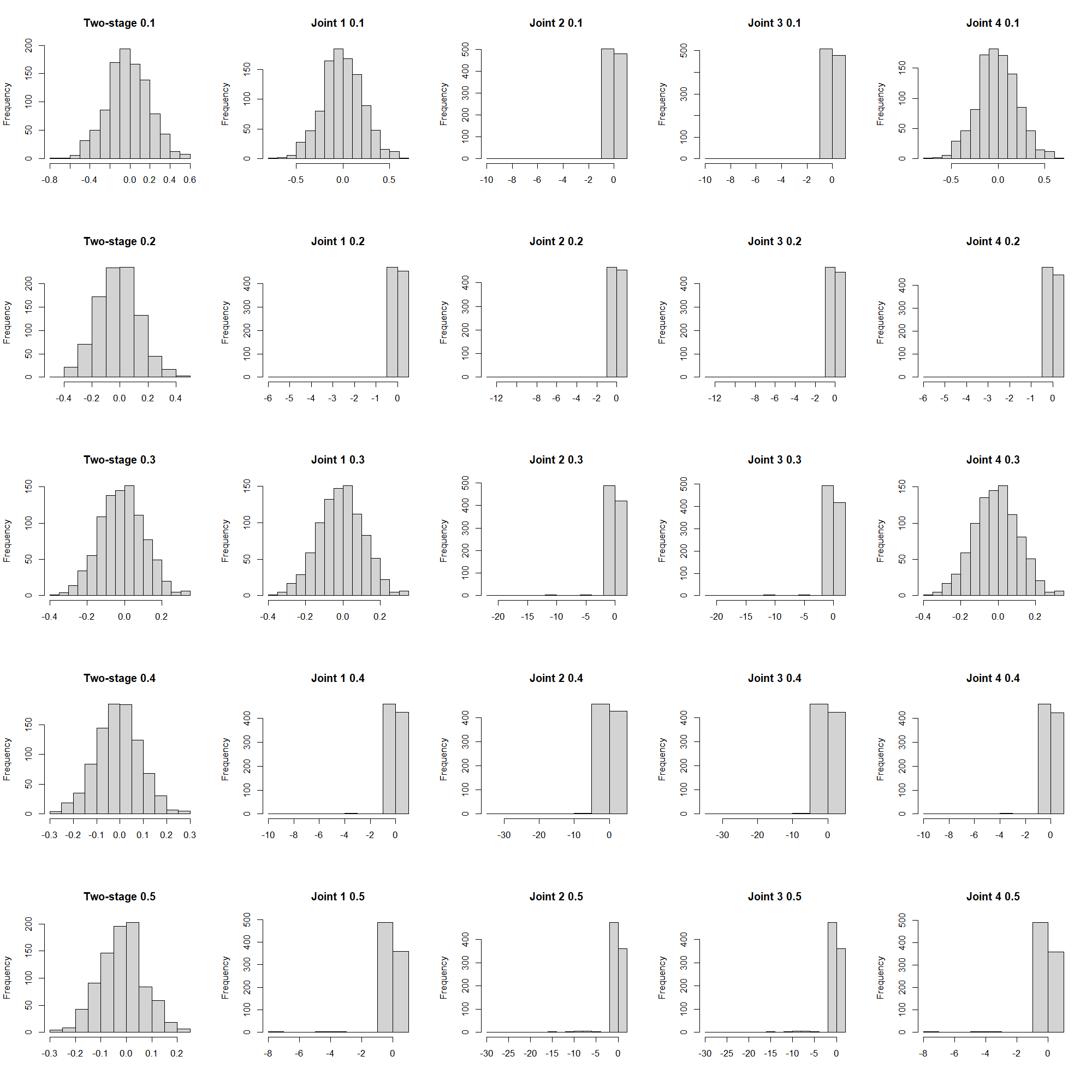}
\caption{The estimated complier expected shortfall treatment effect, $\hat{\gamma}_{1}(\alpha)$, based on  random initial values with sample size $n=3000$. From left to right, each column represents  $\hat{\gamma}_{1}(\alpha)$ from the proposed two-stage method and joint regression method with different specifications of $(G_1(\cdot),G_2(\cdot))$, which are denoted as Joint 1--4. From top to bottom, each row represents  $\hat{\gamma}_{1}(\alpha)$ calculated at $\alpha=\{0.1,\ldots, 0.5\}$.}
\label{SQsimu:f1}
\end{figure}

\section{An Application to the JTPA Dataset}
Job Training Partnership Act (JTPA) is a large publicly-funded training program that began in year 1983. Title II of JTPA, the largest component of JTPA, provides training for economically disadvantaged adults. In this study, applicants were randomized for JTPA trainings in the application process, but did not compel those offered services to participate in training. 
We consider a dataset from Title II of the JTPA study that includes 6102 adult women and 5102 adult men who applied for JTPA between years 1987 and 1989.  The objective is to quantify the effect of JTPA training for the low-income groups.   This dataset has also been considered in \citet{Abadie2002} in estimating  complier quantile treatment effect of JTPA training on the 30-months earnings. 

In this dataset, the response variable is the 30-months earnings, and the observed covariates include race of applicants, whether participants graduated from high-school, including general educational diploma (GED) holders,  marital status, age, aid for families with dependent children (AFDC) receipt (for women), whether worked at least 12 weeks in the 12 months preceding random assignment,  the original recommended service strategy (classroom training, on-the-job training (OJT)/job search assistance (JSA), and others), and whether the data are from the second follow-up survey. The summary statistics of the covariates for adult women and men are shown in Table \ref{SQrd:t1}. 
From Table \ref{SQrd:t1}, we see that only around 60\% percent of participants who are offered the JTPA services actually receive the JTPA training. Besides, the percentage of participants with high school degree in the trainee group is higher than that of the non-trainee group. These imbalances suggest the existence of some unmeasured confounders that affect both the JTPA training status and potential earnings. Without considering these confounders, the standard method may fail to provide a valid estimate of  the effect of JTPA training among low-income subpopulation. 

We first estimate the complier proportions for men and women group, $p_c$, by
\[
p_c=1-\frac{\sum_{i=1}^nD_i(1-V_i)}{\sum_{i=1}^n(1-V_i)}-\frac{\sum_{i=1}^n(1-D_i)V_i}{\sum_{i=1}^nV_i}.
\]
Specifically, the estimated complier proportions for men and women are $0.61$ and $0.64$,  respectively. 
We apply the proposed CRESTE method to estimate the effect of JTPA training for low-income adult men and women separately, using the offer of JTPA service as the instrumental variable. 
As discussed in \cite{Abadie2002}, the validity of using JTPA service as an instrumental variable is guaranteed by the fact that the training offer is randomly assigned. 
For both of the analyses, we adjust for the following covariates: dummy variables for black and Hispanic applicants, an indicator variable for high-school graduates (including GED holders), an indicator for married applicants, 5 age-group indicator variables (22-25, 26-29, 30-35, 36-44, 45-54), an indicator for AFDC receipt (for women), an indicator variable recording 
whether the applicant worked at least 12 weeks in the 12 months preceding random assignment, dummy variables for the original recommended service strategy (classroom, OJT/JSA, other), and an indicator for whether earnings data are from the second follow-up survey.

\begin{table}[!htp]
	\fontsize{8}{9}\selectfont
	\centering
	\caption{Descriptive statistics of variables for participants in JTPA dataset, overall and stratified by the JTPA training status for men and women.}
	\label{SQrd:t1}

\begin{tabular}{|c|c|c|c|c|c|c|}
    \hline
 \multirow{2}{*}{Gender}&\multirow{2}{*}{Variable}&\shortstack{Entire Sample\\ ($N=5102$)} &\multicolumn{2}{c|}{Assignment}&\multicolumn{2}{c|}{Treatment} \\
 \cline{4-7}&&&\shortstack{Training\\ ($N=3339$)}&\shortstack{Non-training \\($N=1703$)}&\shortstack{Trainee\\($N=2136$)}&\shortstack{Non-trainee \\($N=2966$)}\\
\hline
  \multirow{7}{*}{Men}&\shortstack{Training} &0.42&0.62&0.01 &-&-\\
  \cline{2-7}
  &\shortstack{High School or GED} &0.69&0.60&0.69 &0.71&0.68 \\
    \cline{2-7}
  & \shortstack{Age} &32.91 &32.85&33.04&32.76&33.02 \\
     \cline{2-7}
  & \shortstack{Married}&0.35 &0.36&0.34&0.37&0.34  \\
    \cline{2-7}
  &  \shortstack{Black} &0.25&0.25&0.25 &0.26&0.25  \\
    \cline{2-7}
   &\shortstack{Hispanic}&0.10&0.10&0.09 &0.10&0.09  \\
   \cline{2-7}
   &\shortstack{Worked less than 13 weeks}&0.40&0.40&0.40 &0.40&0.40  \\
   \cline{2-7}
      &\shortstack{Data from\\ the second follow-up survey
}&0.29&0.30&0.28 &0.30&0.29  \\
 \cline{2-7}
 &\shortstack{classroom training
}&0.20&0.21&0.19 &0.26&0.16  \\
    \cline{2-7}
     &\shortstack{OJT/JSA
}&0.50&0.50&0.50 &0.46&0.53  \\
 \hline
  \multirow{2}{*}{Gender}&\multirow{2}{*}{Variable}&\shortstack{Entire Sample\\ ($N=6102$)} &\multicolumn{2}{c|}{Assignment}&\multicolumn{2}{c|}{Treatment} \\
 \cline{4-7}&&&\shortstack{Training\\ ($N=4088$)}&\shortstack{Non-training \\($N=2014$)}&\shortstack{Trainee\\($N=2722$)}&\shortstack{Non-trainee \\($N=3380$)}\\
\hline
  \multirow{8}{*}{Women}&\shortstack{Training} &0.45&0.67&0.02 &-&-\\
    \cline{2-7}
  &\shortstack{High School or GED} &0.72&0.73&0.70 &0.75&0.70 \\
     \cline{2-7}
  & \shortstack{Age} &33.33 &33.33&33.35&33.11&33.52 \\
    \cline{2-7}
   &\shortstack{Married}&0.22 &0.22&0.21&0.22&0.21  \\
     \cline{2-7}
   & \shortstack{Black} &0.26&0.27&0.26 &0.26&0.27  \\
    \cline{2-7}
   &\shortstack{Hispanic}&0.12&0.12&0.12 &0.12&0.11  \\
     \cline{2-7}
   &\shortstack{Worked less than 13 weeks}&0.52&0.52&0.52 &0.51&0.53  \\
    \cline{2-7}
     &\shortstack{Data from\\ the second follow-up survey
}&0.26&0.26&0.25 &0.26&0.25  \\
 \cline{2-7}
 &\shortstack{classroom training
}&0.38&0.38&0.39 &0.45&0.33  \\
    \cline{2-7}
     &\shortstack{OJT/JSA
}&0.37&0.37&0.38 &0.32&0.42  \\
    \cline{2-7}
    & \shortstack{AFDC}&0.31&0.30&0.31&0.32&0.30 \\
    \hline
\end{tabular}
\end{table}

We perform our analysis for adult men and women  separately at $\alpha=0.25$ and $0.5$, which are the groups of men or women whose 30-months earnings are below the first  quartile and the  median of the 30-months earnings (after adjusting for the covariates), respectively.  We apply the proposed method and the as-treated method, which refers to applying the two-stage method of \cite{Barendse2020} to the entire dataset. 
  We use the same variables as those of \cite{Abadie2002} to estimate $\kappa_{v,i}$, where $\pi$ is estimated by the empirical estimator of $\mathbb{E}(V)$. We use 30-month earnings to estimate $v$ for adult men, and 30-month earnings and the classroom training indicator to estimate $v$ for women. We adopt the fourth-order Epanechnikov kernel in $\hat{v}_i$, and the bandwidth is selected via cross-validation. In our analysis, the selected bandwidths  are $5200$ and $7800$ for adult men and women, respectively.

Table~\ref{SQrd:t2} presents the estimated CQTE and CRESTE (with standard errors computed via the bootstrap) of the effect of JTPA services on earnings for both men and women subgroups.  
From Table~\ref{SQrd:t2}, we see that the proposed method concludes that JTPA training does not have a statistically significant effect on improving earnings for low-income men. In contrast, the as-treated method overestimates the effect of JTPA training on earnings and the conclusion obtained from the as-treated method can be misleading.  
From the discussion in \citet{Abadie2002}, the difference between these two methods  may be due to self-selection  or an effort by program operators to exclude men with low earnings potentials from  JTPA training.  
In contrast,  both results from the as-treated analysis and the proposed method show that the JTPA service has significant effect on improving earnings for low-income adult women, even though overestimation of the effects from the as-treated analysis remain visible. Quantitatively we note that for adult women in the lower half of their earnings (after adjusted for the covariates), the average JTPA training effect is about \$1086, much lower than the median effect of \$1760. It is notable that  CQTE and CRESTE at the $0.25$ quantile level do not differ  much for women, which indicates that  we do not have a significant spread toward lower values of the response below the first quartile for women.
The difference between CQTE and CRESTE at $\alpha = 0.25$ is greater for men than that of women, which suggests a greater lower tail spread for the effects of JTPA training in terms of the 30-month earnings for men than for women.

\begin{table}[!htp]
	\fontsize{8}{9}\selectfont
	\centering
	\caption{ The estimated coefficients (Est) and the corresponding estimated standard error (Boot SD) of JTPA trainings from the proposed CRESTE method and the as-treated method in JTPA dataset at $\alpha=0.25$ and $0.5$.}
	\label{SQrd:t2}

\begin{tabular}{|c|c|c|c|c|c|c|}
    \hline
 \multirow{2}{*}{Gender}&\multirow{2}{*}{$\alpha$}& &\multicolumn{2}{c|}{CRESTE}&\multicolumn{2}{c|}{As-treated} \\
 \cline{4-7}&&&$\beta_{1}(\alpha)$&$\gamma_{1}(\alpha)$&$\beta_{1}(\alpha)$&$\gamma_{1}(\alpha)$\\
\hline
     \multirow{4}{*} {Men}& \multirow{2}{*}{$0.25$}&Est&755&383&2510&1502\\
      \cline{3-7} &&  SD&588&408&397&211\\
        \cline{2-7}
    & \multirow{2}{*}{$0.5$}&Est&1663&921&4420&2984\\
      \cline{3-7} &&   SD&986&617&673&358\\
     \hline
     \multirow{4}{*} {Women}& \multirow{2}{*}{$0.25$}&Est&657&633&1013&792\\
      \cline{3-7} &&   SD&222&187&177&129\\
          \cline{2-7}
    & \multirow{2}{*}{$0.5$}&Est&1760&1086&2707&1633\\
     \cline{3-7} &&   SD&606&323&427&221\\ 
     \hline
\end{tabular}
\end{table}

\section{Discussion}
We consider estimating the CRESTE, i.e., the expected shortfall treatment effect for the compliers, in observational studies.  Different from the quantile treatment effect, the expected shortfall treatment effect measures the aggregate quantile treatment effect over the lower (or upper) tail of the conditional distribution for the response variable. Moreover, the average treatment effect can be treated as a special case of the expected shortfall treatment effect with $\alpha=1$. 

We propose a two-stage method to estimate CRESTE by utilizing the special characteristic of a binary instrumental variable. 
 The proposed two-stage estimation procedure involves solving a weighted quantile regression to obtain the thresholding quantile at the first stage, and solving a weighted least squares problem to estimate the expected shortfall parameter at the second stage.  The method can be readily implemented using existing software such as the \texttt{R} package \texttt{quantreg} for solving a weighted quantile regression, and the \texttt{lm} function in \texttt{base} \texttt{R} for solving the weighted least squares problem.    
 Compared to the approach in \citet{Chen2021} that requires solving a non-convex loss function, we demonstrate in numerical studies that the proposed two-stage method has a more stable solution and is computationally efficient. 
Theoretically, the estimated CRESTE from the proposed two-stage method is locally robust to the estimates of quantiles at the first stage due to Neyman-orthogonalization in the relevant expected shortfall score function. 

\section*{Acknowledgement}
\noindent  We thank Le-Yu Chen and Yu-Min Yen for sharing \texttt{R} code to estimate the complied expected shortfall treatment effect using the joint regression approach.

\appendix
\section{Propositions}
\label{SQpropappen}

\begin{prop}
Under Assumptions (A1) and (A2), Models~\eqref{SQeq1} and~\eqref{SQeq2} are equivalent. 
\label{SQprop2}
\end{prop}
\begin{proof}
Given the fact that $Y=D\times Y_1+(1-D)\times Y_0$, we have 
\[
\begin{split}
Q_{\alpha}(Y|D=1,\bX,D_1>D_0)&=\inf\{y:\Pr(Y\le y|D=1,\bX,D_1>D_0)\ge \alpha\}\\
&=\inf\{y:\Pr(Y_1\le y|D=1,\bX,D_1>D_0)\ge \alpha\}.
\end{split}
\]
Since $V=D$ for the compliers, we have $\inf\{y:\Pr(Y_1\le y|D=1,\bX,D_1>D_0)\ge \alpha\}
=\inf\{y:\Pr(Y_1\le y|V=1,\bX,D_1>D_0)\ge \alpha\}$. Under Assumptions (A1) and (A2), we have 
\[
\begin{split}
\inf\{y:\Pr(Y_1\le y|V=1,\bX,D_1>D_0)\ge \alpha\}=&\inf\{y:\Pr(Y_1\le y|\bX,D_1>D_0)\ge \alpha\}\\
=&Q_{\alpha}(Y_1|\bX,D_1>D_0).
\end{split}
\]
Similarly, we have $Q_{\alpha}(Y|D=0,\bX,D_1>D_0)=Q_{\alpha}(Y_0|\bX,D_1>D_0)$.

The aforementioned equations imply that $Q_{\alpha}(Y|D=d,\bX,D_1>D_0)=Q_{\alpha}(Y_d|\bX,D_1>D_0)$ for $d=\{0,1\}$.  Similarly, we have
\[
\begin{split}
S_{\alpha}(Y|D=d,\bX,D_1>D_0)\doteq&\frac{1}{\alpha}\int_0^\alpha Q_{u}(Y|D=d,\bX,D_1>D_0)du\\
=&\frac{1}{\alpha}\int_0^\alpha Q_{u}(Y_d|\bX,D_1>D_0)du\\
\doteq&S_{\alpha}(Y_d|\bX,D_1>D_0).
\end{split}
\]
Thus, Model \eqref{SQeq1} and Model \eqref{SQeq2} are equivalent.
\end{proof}

\begin{prop}
Under Assumptions (A1) and (A2), in the one-sided compliance case where  subjects with $V=0$ have no access to the treatment  (i.e., $\Pr(D_0=0|\bX)=1$), $\gamma_{1}(\alpha)=S_{\alpha}(Y_1|\bX,D=1)-S_{\alpha}(Y_0|\bX,D=1)$.
\label{SQprop1}
\end{prop}

\begin{proof}

From the proof of Proposition \eqref{SQprop2}, we have $S_{\alpha}(Y|D=d,\bX,D_1>D_0)=S_{\alpha}(Y_d|\bX,D_1>D_0)$. Thus, 
\[
\begin{split}
\gamma_{1}(\alpha)\doteq& S_{\alpha}(Y_1|\bX,D_1>D_0)-S_{\alpha}(Y_0|\bX,D_1>D_0)\\
=&S_{\alpha}(Y|D=1,\bX,D_1>D_0)-S_{\alpha}(Y|D=0,\bX,D_1>D_0).
\end{split}
\]
Note that $\Pr(D_0=0|\bX)=1$  implies that subjects with $D=1$ must belong to the complier group. Given the fact that $Y=D\times Y_1+(1-D)\times Y_0$, we have 
\[
\begin{split}
Q_{\alpha}(Y|D=1,\bX,D_1>D_0)&=\inf\{y:\Pr(Y_1\le y|D=1,\bX,D_1>D_0)\ge \alpha\}\\
&=\inf\{y:\Pr(Y_1\le y|D=1,\bX)\ge \alpha\}\\
&=Q_{\alpha}(Y_1|\bX,D=1).
\end{split}
\]
Moreover, 
\[
\begin{split}
Q_{\alpha}(Y|D=0,\bX,D_1>D_0)=&\inf\{y:\Pr(Y_0\le y|D=0,\bX,D_1>D_0)\ge \alpha\}\\
=&\inf\{y:\Pr(Y_0\le y|V=0,\bX,D_1=1)\ge \alpha\}\\
=&\inf\{y:\Pr(Y_0\le y|V=1,\bX,D_1=1)\ge \alpha\}\\
=&\inf\{y:\Pr(Y_0\le y|D=1,\bX)\ge \alpha\}\\
=&Q_{\alpha}(Y_0|\bX,D=1),
\end{split}
\]
where the second equality follows from the one-sided compliance constraint, $\Pr(D_0=0|\bX)=1$, and the third equality is ensured by Assumption (A1). Assumption (A1) along with $\Pr(D_0=0|\bX)=1$ further imply the fourth equality. 

For $d=\{0,1\}$, we have
\[
\begin{split}
S_{\alpha}(Y_d|D=1,\bX)=&\frac{1}{\alpha}\int_0^\alpha Q_{u}(Y_d|D=1,\bX)du\\
=&\frac{1}{\alpha}\int_0^\alpha Q_{u}(Y|D=d,\bX,D_1>D_0)du\\
=&S_{\alpha}(Y|D=d,\bX,D_1>D_0).
\end{split}
\]
Thus, $\gamma_{1}(\alpha)\doteq S_{\alpha}(Y_1|\bX,D_1>D_0)-S_{\alpha}(Y_0|\bX,D_1>D_0)$ is equivalent to $S_{\alpha}(Y_1|\bX,D=1)-S_{\alpha}(Y_0|\bX,D=1)$.
\end{proof}

\begin{prop}
Let  $\kappa_v(Y,\bZ)=\Pr(D_1>D_0|Y,\bZ)$. Under Assumptions (A1)--(A4), we have  
\[
\kappa_v(Y,\bZ)=1-\frac{D\{1-v(Y,\bZ)\}}{1-\pi(\bX)}-\frac{(1-D)v(Y,\bZ)}{\pi(\bX)},
\] 
where $v(Y,\bZ)=\Pr(V=1|Y,\bZ)$ and $\pi(\bX)=\Pr(V=1|\bX)$.
\label{SQprop3}
\end{prop}

\begin{proof}
Note that  $D(1-V)$ only differs from zero when $D=1$ and $V=0$.  By the monotonicity assumption, $D_0=1$ implies $D_1=1$. Then
\[
\begin{split}
&\mathbb{E}\{D(1-V)|Y,\bZ\}=\Pr\{D(1-V)=1|Y,\bZ\}=\Pr(D_1=D_0=1,V=0|Y,\bZ)\\
&=\Pr(D_1=D_0=1|Y,\bZ)\Pr(V=0|D_1=D_0=1,Y_1,\bX)\\
&=\Pr(D_1=D_0=1|Y,\bZ)\Pr(V=0|\bX),
\end{split}
\]
where the last equality follows from Assumptions (A1) and (A4), i.e.,  $V$ is independent of $(D_1,D_0,Y_1,Y_0)$ conditional on $\bX$. 
Similarly, we can show that $\mathbb{E}\{(1-D)V|Y,\bZ\}=\Pr(D_1=D_0=0|Y,\bZ)\Pr(V=1|\bX)$. 
Therefore,
\[
\begin{split}
&1-\frac{D\{1-v(Y,\bZ)\}}{1-\pi(\bX)}-\frac{(1-D)v(Y,\bZ)}{\pi(\bX)}\\
&=\mathbb{E}\Big\{1-\frac{D(1-V)}{\Pr(V=0|\bX)}-\frac{(1-D)V}{\Pr(V=1|\bX)}|Y,\bZ\Big\}\\
&=1-\Pr(D_1=D_0=1|Y,\bZ)-\Pr(D_1=D_0=0|Y,\bZ)\\
&=\Pr(D_1>D_0|Y,\bZ),
\end{split}
\]
where the last equality follows from the monotonicity assumption.
\end{proof}

\section{Proof of Theorems \ref{SQthm0}--\ref{SQthm2} \label{AppB}}
We start with some lemmas that will be helpful in proving Theorems~\ref{SQthm0}--\ref{SQthm2}. The proofs of the Lemmas are deferred to  ~\ref{appendix:C}.

\begin{lemma}
Under Conditions (C1)--(C8), we have
\[
n^{-1/2}\sum_{i=1}^n\tilde{\kappa}_{v,i}\bZ_ig_{\alpha,i}(\bgamma^\ast,\bbeta^\ast)=n^{-1/2}\sum_{i=1}^n\bPsi_i(\alpha)+o_p(1),
\]
where $\bPsi_i(\alpha)=\bfm_2(Y_i,\bZ_i,\alpha)\{1-\frac{D_i(1-V_i)}{1-\pi_i}-\frac{(1-D_i)V_i}{\pi_i}\}+\bH_2(\bX_i,\alpha)\{V_i-\pi(\bX_i)\}$, with
$\bfm_2(Y_i,\bZ_i,\alpha)=\bZ_ig_{\alpha,i}(\bgamma^\ast,\bbeta^\ast)$ and $\bH_2(\bX_i,\alpha)=\mathbb{E}[\bfm_2(Y_i,\bZ_i,\alpha)\{\frac{(1-D_i)v_i}{\pi_i^2}-\frac{D_i(1-v_i)}{(1-\pi_i)^2}\}|\bX_i]$.
\label{SQlemma1}
\end{lemma}

\begin{lemma}
Under Conditions (C1)--(C8), we have
\[
n^{-1/2}\sum_{i=1}^n\tilde{\kappa}_{v,i}\bZ_i\{\alpha-I(Y_i<\bZ_i^\top\bbeta^\ast)\}=n^{-1/2}\sum_{i=1}^n\bPhi_i(\alpha)+o_p(1),
\]
where $\bPhi_i(\alpha)=\bfm_1(Y_i,\bZ_i,\alpha)\{1-\frac{D_i(1-V_i)}{1-\pi_i}-\frac{(1-D_i)V_i}{\pi_i}\}+\bH_1(\bX_i,\alpha)\{V_i-\pi(\bX_i)\}$. Here 
$\bfm_1(Y_i,\bZ_i,\alpha)=\bZ_i\{\alpha-I(Y_i<\bZ_i^\top\bbeta^\ast)\}$ and $\bH_1(\bX_i,\alpha)=\mathbb{E}[\bfm_1(Y_i,\bZ_i,\alpha)\{\frac{(1-D_i)v_i}{\pi_i^2}-\frac{D_i(1-v_i)}{(1-\pi_i)^2}\}|\bX_i]$.
\label{SQlemma2}
\end{lemma}


\subsection{Proof of Theorem \ref{SQthm0}}
\begin{proof}
This proof of Theorem~\ref{SQthm0} is adapted from the proof of Theorem 3.1 in \cite{Abadie2002}, which we detail here for completeness.   Let $L_n(\bzeta,\kappa)=\sum_{i=1}^n l_i(\bzeta,\kappa)$ with
\[
l_i(\bzeta,\kappa)=\kappa_{v,i}\{\rho_{\alpha}(\epsilon_{i}-n^{-1/2}\bZ_i^\top\bzeta)-\rho_{\alpha}(\epsilon_{i})\},
\]
where $\epsilon_i=Y_i-\bZ_i^\top\bbeta^\ast$ and $\rho_{\alpha}(u)=u\{\alpha-I(u\le 0)\}$.  Similarly, let  $L_n(\bzeta,\tilde{\kappa})=\sum_{i=1}^n l_i(\bzeta,\tilde{\kappa})$ with
$l_i(\bzeta,\tilde{\kappa})=\tilde{\kappa}_{v,i}\{\rho_{\alpha}(\epsilon_{i}-n^{-1/2}\bZ_i^\top\bzeta)-\rho_{\alpha}(\epsilon_{i})\}.$

The function $L_n(\bzeta,\tilde{\kappa})$ is a convex function with respect to $\bzeta$ and a minimizer takes the form $\bzeta_n=n^{1/2}(\hat{\bbeta}-\bbeta^\ast)$. Besides,
\[
\frac{\partial l_i(\bzeta,\kappa)}{\partial \bzeta}=-n^{-1/2}\bZ_i \kappa_{v,i}\{\alpha-I(\epsilon_i-n^{-1/2}\bZ_i^\top\bzeta\le 0)\},
\]
almost surely. Denote $l(\bzeta,\kappa)=\kappa_{v}\{\rho_{\alpha}(\epsilon-n^{-1/2}\bZ^\top\bzeta)-\rho_{\alpha}(\epsilon)\}$, 
where $\epsilon=Y-\bZ^\top\bbeta^\ast$. By Condition (C4) and Weierstrass domination, we have 
\[
\frac{\partial \mathbb{E}\{l(\bzeta,\kappa)\}}{\partial \bzeta}\Big|_{\bzeta=\bzero}=-n^{-1/2}\mathbb{E}[\bZ \kappa_{v}\{\alpha-I(\epsilon\le 0)\}]=0,
\]
and 
\[
\begin{split}
\frac{\partial \mathbb{E}\{l(\bzeta,\kappa)\}}{\partial \bzeta\partial \bzeta^\top}\Big|_{\bzeta=\bzero}=&-n^{-1}\mathbb{E}\{f(\bZ^\top\bbeta^\ast|\bZ,D_1>D_0)\bZ\bZ^\top|D_1>D_0\}\Pr(D_1>D_0)\\
=&n^{-1}\mathbb{E}\{I(D_1>D_0)f(\bZ^\top\bbeta^\ast|\bZ,D_1>D_0)\bZ\bZ^\top\}\\
=&n^{-1}\bJ_1.
\end{split}
\]
From Conditions (C3) and (C4), $\bJ_1$ is nonsingular. By a Taylor's  expansion, we have
\begin{equation}
\mathbb{E}\{L_n(\bzeta,\kappa)\}=\frac{1}{2}\bzeta^\top\bJ_1\bzeta+o_p(1).
\label{SQthm0e0}
\end{equation}

To prove Theorem \ref{SQthm0}, we  first show
\begin{equation}
L_n(\bzeta,\tilde{\kappa})=\mathbb{E}\{L_n(\bzeta,\kappa)\}-n^{-1/2}\sum_{i=1}^n\tilde{\kappa}_{v,i}\bZ_i^\top\bzeta\{\alpha-I(\epsilon_i\le 0)\}+o_p(1).
\label{SQthm0e1}
\end{equation}
Note that
\[
\begin{split}
L_n(\bzeta,\tilde{\kappa})=&\mathbb{E}\{L_n(\bzeta,\kappa)\}+[L_n(\bzeta,\tilde{\kappa})-\mathbb{E}\{L_n(\bzeta,\kappa)\}]\\
=&\mathbb{E}\{L_n(\bzeta,\kappa)\}-n^{-1/2}\sum_{i=1}^n\tilde{\kappa}_{v,i}\bZ_i^\top\bzeta\{\alpha-I(\epsilon_i\le 0)\}\\
&+\Big[L_n(\bzeta,\tilde{\kappa})+n^{-1/2}\sum_{i=1}^n\tilde{\kappa}_{v,i}\bZ_i^\top\bzeta\{\alpha-I(\epsilon_i\le 0)\}-\mathbb{E}\{L_n(\bzeta,\kappa)\}\Big].
\end{split}
\]
For notational convenience, let  $U_n(\bZ_i,\kappa,\bzeta)=l_i(\bzeta,\kappa)+n^{-1/2}\kappa_{v,i}\bZ_i^\top\bzeta\{\alpha-I(\epsilon_i\le 0)\}$ and let $U_n(\bZ_i,\tilde{\kappa},\bzeta)=l_i(\bzeta,\tilde{\kappa})+n^{-1/2}\tilde{\kappa}_{v,i}\bZ_i^\top\bzeta\{\alpha-I(\epsilon_i\le 0)\}$.  Then, 
\[
\begin{split}
L_n(\bzeta,\tilde{\kappa})=&\mathbb{E}\{L_n(\bzeta,\kappa)\}+[L_n(\bzeta,\tilde{\kappa})-\mathbb{E}\{L_n(\bzeta,\kappa)\}]\\
=&\mathbb{E}\{L_n(\bzeta,\kappa)\}-n^{-1/2}\sum_{i=1}^n\tilde{\kappa}_{v,i}\bZ_i^\top\bzeta\{\alpha-I(\epsilon_i\le 0)\}\\
&+\sum_{i=1}^n U_n(\bZ_i,\tilde{\kappa},\bzeta)-\sum_{i=1}^n U_n(\bZ_i,\kappa,\bzeta)\\
&+\sum_{i=1}^n U_n(\bZ_i,\kappa,\bzeta)-\mathbb{E}\{U_n(\bZ_i,\kappa,\bzeta)\}.
\end{split}
\]
The last equation holds because $\mathbb{E}[\kappa_{v,i}\bZ_i\{\alpha-I(\epsilon_i\le 0)\}]=0$.

Let $f_{\epsilon|\bZ}(\cdot)$ and $F_{\epsilon|\bZ}(\cdot)$ be the conditional density function and cumulative distribution function of $\epsilon$ given $\bZ$, respectively. Since $\epsilon=Y-\bZ^\top\bbeta^\ast$ and 
\[
|\rho_{\alpha}(\epsilon_{i}-n^{-1/2}\bZ_i^\top\bzeta)-(\epsilon_{i}-n^{-1/2}\bZ_i^\top\bzeta)\{\alpha-I(\epsilon_i\le 0)\}|\le I(|\epsilon_{i}|\le |n^{-1/2}\bZ_i^\top\bzeta|)\cdot|n^{-1/2}\bZ_i^\top\bzeta|,
\]
we have 
\begin{equation}
\begin{split}
&\mathbb{E}[n\cdot|\rho_{\alpha}(\epsilon_{i}-n^{-1/2}\bZ_i^\top\bzeta)-(\epsilon_{i}-n^{-1/2}\bZ_i^\top\bzeta)\{\alpha-I(\epsilon_i\le 0)\}|]\\
\le& \mathbb{E}[I(|\epsilon_{i}|\le |n^{-1/2}\bZ_i^\top\bzeta|)\cdot|n^{1/2}\bZ_i^\top\bzeta|]\\
=&\mathbb{E}\Big[\frac{F_{\epsilon|\bZ}(|n^{-1/2}\bZ_i^\top\bzeta|)-F_{\epsilon|\bZ}(-|n^{-1/2}\bZ_i^\top\bzeta|)}{n^{-1/2}}|\bZ_i^\top\bzeta|\Big]\\
\to& 2\mathbb{E}\Big[f_{\epsilon|\bZ}(0)|\bZ_i^\top\bzeta|^2\Big]<\infty.
\end{split}
\label{SQthm0e2}
\end{equation}

Then,
\begin{equation}
\begin{split}
&|\sum_{i=1}^n\bU_n(\bZ_i,\tilde{\kappa},\bzeta)-U_n(\bZ_i,\kappa,\bzeta)|\\
\le&\sum_{i=1}^n|\tilde{\kappa}_{v,i}-\kappa_{v,i}|\cdot|\{\rho_{\alpha}(\epsilon_{i}-n^{-1/2}\bZ_i^\top\bzeta)-(\epsilon_{i}-n^{-1/2}\bZ_i^\top\bzeta)\{\alpha-I(\epsilon_i\le 0)\}\}|\\
\le&\sup_i|\tilde{\kappa}_{v,i}-\kappa_{v,i}|\frac{1}{n}\sum_{i=1}^n n\cdot|\rho_{\alpha}(\epsilon_{i}-n^{-1/2}\bZ_i^\top\bzeta)-(\epsilon_{i}-n^{-1/2}\bZ_i^\top\bzeta)\{\alpha-I(\epsilon_i\le 0)\}|.
\end{split}
\label{SQthm0e2_1}
\end{equation}

Besides, from Conditions (C7) and (C8), Lemma B.3 of \cite{Newey1994} implies $\sup_i|\hat{v}_i-v|=o_p(n^{-1/4})$ and $\sup_i|\hat{\pi}_i-\pi|=o_p(n^{-1/4})$. Then by Condition (C5), we have $\sup_i|\hat\kappa_{v,i}-\kappa_{v,i}|=o_p(n^{-1/4})$. From Condition (C6), we can obtain 
\begin{equation}
\sup_i|\tilde{\kappa}_{v,i}-\kappa_{v,i}|=o_p(n^{-1/4}).
\label{SQthm1:e1}
\end{equation}

Coupled with Equation~\eqref{SQthm0e2} and $\sup_i|\tilde{\kappa}_{v,i}-\kappa_{v,i}|=o_p(n^{-1/4})$, Equation \eqref{SQthm0e2_1} implies that
\begin{equation}
|\sum_{i=1}^n\bU_n(\bZ_i,\tilde{\kappa},\bzeta)-U_n(\bZ_i,\kappa,\bzeta)|=o_p(1).
\label{SQthm0e3}
\end{equation}
Besides,
\begin{equation}
\begin{split}
\mathbb{E}\Big[\Big(\sum_{i=1}^n\bU_n(\bZ_i,\kappa,\bzeta)-\mathbb{E}\{U_n(\bZ_i,\kappa,\bzeta)\}\Big)^2\Big]&\le \sum_{i=1}^n \mathbb{E}\{U_n(\bZ_i,\kappa,\bzeta)^2\}\\
&\le \mathbb{E}\{I(|\epsilon_{i}|\le |n^{-1/2}\bZ_i^\top\bzeta|)\cdot|\bZ_i^\top\bzeta|^2\}\\
&\to 0,
\end{split}
\label{SQthm0e4}
\end{equation}
where the first inequality follows from the cancellation of cross-product terms.
Based on Equations~\eqref{SQthm0e3} and \eqref{SQthm0e4}, we have
\[L_n(\bzeta,\tilde{\kappa})=\mathbb{E}\{L_n(\bzeta,\kappa)\}-n^{-1/2}\sum_{i=1}^n\tilde{\kappa}_{v,i}\bZ_i^\top\bzeta\{\alpha-I(\epsilon_i\le 0)\}+o_p(1).
\]

Next, we will show $\bzeta_n=n^{1/2}(\hat\bbeta-\bbeta^\ast)\to_d N(\bzero, \bJ_1^{-1}\bOmega_1\bJ_1^{-1})$. From Equations~\eqref{SQthm0e0} and \eqref{SQthm0e1}, for a given $\bzeta$,
\[
\begin{split}
L_n(\bzeta,\tilde{\kappa})=&\mathbb{E}\{L_n(\bzeta,\kappa)\}-n^{-1/2}\sum_{i=1}^n\tilde{\kappa}_{v,i}\bZ_i^\top\bzeta\{\alpha-I(\epsilon_i\le 0)\}+o_p(1)\\
=&\frac{1}{2}\bzeta^\top\bJ_1\bzeta-n^{-1/2}\sum_{i=1}^n\tilde{\kappa}_{v,i}\bZ_i^\top\bzeta\{\alpha-I(\epsilon_i\le 0)\}+o_p(1).
\end{split}
\]
Since $L_n(\bzeta,\tilde{\kappa})+n^{-1/2}\sum_{i=1}^n\tilde{\kappa}_{v,i}\bZ_i^\top\bzeta\{\alpha-I(\epsilon_i\le 0)\}$ is convex in $\bzeta$. From Pollard's convexity lemma \citep{Pollard1991}, for any compact subset $\mathcal{T}\subset\mathbb{R}^{l+1}$,  we have
\begin{equation}
\sup_{\bzeta\in\mathcal{T}}|L_n(\bzeta,\tilde{\kappa})+n^{-1/2}\sum_{i=1}^n\tilde{\kappa}_{v,i}\bZ_i^\top\bzeta\{\alpha-I(\epsilon_i\le 0)\}-\frac{1}{2}\bzeta^\top\bJ_1\bzeta|=o_p(1).
\label{SQthm0e5}
\end{equation}
Let $\bfeta_n=\bJ_1^{-1}n^{-1/2}\sum_{i=1}^n\tilde{\kappa}_{v,i}\bZ_i\{\alpha-I(\epsilon_i\le 0)\}$. Note that
\[
\frac{1}{2}(\bzeta-\bfeta_n)^\top\bJ_1(\bzeta-\bfeta_n)=\frac{1}{2}\bzeta^\top\bJ_1\bzeta-n^{-1/2}\sum_{i=1}^n\tilde{\kappa}_{v,i}\bZ_i^\top\bzeta\{\alpha-I(\epsilon_i\le 0)\}+\frac{1}{2}\bfeta_n^\top\bJ_1\bfeta_n.
\]
From Equation~\eqref{SQthm0e5}, for any compact subset $\mathcal{T}\subset\mathbb{R}^{l+1}$, we have
\[
\sup_{\bzeta\in\mathcal{T}}|L_n(\bzeta,\tilde{\kappa})-\frac{1}{2}(\bzeta-\bfeta_n)^\top\bJ_1(\bzeta-\bfeta_n)+\frac{1}{2}\bfeta_n^\top\bJ_1\bfeta_n|=o_p(1).
\]

Finally, by an application of  Lemma 3 in \cite{Buchinsky1998}, we have $\bzeta_n=\bfeta_n+o_p(1)$. From Lemma \ref{SQlemma2}, we have 
\[
n^{1/2}(\hat\bbeta-\bbeta^\ast)\to_d N(\bzero, \bJ_1^{-1}\bOmega_1\bJ_1^{-1})
\]
as desired.
\end{proof}

\subsection{Proof of Theorem \ref{SQthm1}}
\begin{proof}

From Theorem \ref{SQthm0}, we have $\|\hat\bbeta-\bbeta^\ast\|\to_p 0$. By the definition of $\bgamma^\ast$, $\mathbb{E}\{I(D_1>D_0)\bZ g_{\alpha}(\bZ,\bgamma^\ast,\bbeta^\ast)\}=\bzero$. For any $\bb\in\mathcal{B}$ and $\bb\neq \bgamma^\ast$, we have
\[
\begin{split}
&\mathbb{E}\{I(D_1>D_0)\bZ g_{\alpha}(\bZ,\bb,\bbeta^\ast)\}-\mathbb{E}\{I(D_1>D_0)\bZ g_{\alpha}(\bZ,\bgamma^\ast,\bbeta^\ast)\}\\
=&\mathbb{E}[I(D_1>D_0)\bZ(\bZ^\top\bb-\frac{1}{\alpha}(Y-\bZ^\top\bbeta^\ast)I(Y\le \bZ^\top\bbeta^\ast)-\bZ^\top\bbeta^\ast)\\
&-I(D_1>D_0)\bZ(\bZ^\top\bgamma^\ast-\frac{1}{\alpha}(Y-\bZ^\top\bbeta^\ast)I(Y\le \bZ^\top\bbeta^\ast)-\bZ^\top\bbeta^\ast)]\\
=&\mathbb{E}\{I(D_1>D_0)\bZ\bZ^\top\}(\bb-\bgamma^\ast).
\end{split}
\]
Thus, $\mathbb{E}\{I(D_1>D_0)\bZ g_{\alpha}(\bZ,\bb,\bbeta^\ast)\}=\mathbb{E}\{I(D_1>D_0)\bZ g_{\alpha}(\bZ,\bgamma^\ast,\bbeta^\ast)\}=\bzero$ if and only if $\mathbb{E}\{I(D_1>D_0)\bZ\bZ^\top\}(\bb-\bgamma^\ast)=\bzero$. By Condition (C3), i.e., $\mathbb{E}\{I(D_1>D_0)\bZ\bZ^\top\}$ is of full rank, we have $\mathbb{E}\{I(D_1>D_0)\bZ g_{\alpha}(\bZ,\bb,\bbeta^\ast)\}=\bzero$ if and only if $\bb=\bgamma^\ast$. Thus, $\bgamma^\ast$ is the unique solution of  $\mathbb{E}\{I(D_1>D_0)\bZ g_{\alpha}(\bZ,\bb,\bbeta^\ast)\}=\bzero$.

Next, we will show that $n^{-1}\sum_{i=1}^n\tilde{\kappa}_{v,i}\bZ_i g_{\alpha,i}(\bb,\hat{\bbeta})$ uniformly converges to $\mathbb{E}\{\kappa_v\bZ g_{\alpha}(\bZ,\bb,\bbeta^\ast)\}$ for $\bb\in\mathcal{B}$ and $\hat{\bbeta}$ given that $\|\hat{\bbeta}-{\bbeta}^{\ast}\|=o_p(1)$. 
By adding and subtracting terms, it can be shown that 
\begin{equation}
\begin{split}
&n^{-1}\sum_{i=1}^n\tilde{\kappa}_{v,i}\bZ_i g_{\alpha,i}(\bb,\hat{\bbeta})\\
=&\mathbb{E}\{\kappa_v\bZ g_{\alpha}(\bZ,\bb,\bbeta^\ast)\}+n^{-1}\sum_{i=1}^n{\kappa}_{v,i}\bZ_i g_{\alpha,i}(\bb,\bbeta^\ast)-\mathbb{E}\{\kappa_v\bZ g_{\alpha}(\bZ,\bb,\bbeta^\ast)\}\\
&+n^{-1}\sum_{i=1}^n(\tilde{\kappa}_{v,i}-\kappa_{v,i})\bZ_i g_{\alpha,i}(\bb,\hat{\bbeta})+n^{-1}\sum_{i=1}^n{\kappa}_{v,i}\bZ_i \{g_{\alpha,i}(\bb,\hat{\bbeta})-g_{\alpha,i}(\bb,\bbeta^\ast)\}.\\
\end{split}
\label{SQthm1eadd1}
\end{equation}

Let $\mathbf{I}_n(\bb)=n^{-1}\sum_{i=1}^n{\kappa}_{v,i}\bZ_i g_{\alpha,i}(\bb,\bbeta^\ast)-\mathbb{E}\{\kappa_v\bZ g_{\alpha}(\bZ,\bb,\bbeta^\ast)\}$, $\mathbf{II}_n(\bb)=n^{-1}\sum_{i=1}^n(\tilde{\kappa}_{v,i}-\kappa_{v,i})\bZ_i\cdot \\g_{\alpha,i}(\bb,\hat{\bbeta})$, and $\mathbf{III}_n(\bb)=n^{-1}\sum_{i=1}^n{\kappa}_{v,i}\bZ_i \{g_{\alpha,i}(\bb,\hat{\bbeta})-g_{\alpha,i}(\bb,\bbeta^\ast)\}$. To prove that $n^{-1}\sum_{i=1}\tilde{\kappa}_{v,i}\bZ_i\cdot \\g_{\alpha,i}(\bb,\hat{\bbeta})$ uniformly converges to $\mathbb{E}\{\kappa_v\bZ g_{\alpha}(\bZ,\bb,\bbeta^\ast)\}$ for $\bb\in\mathcal{B}$, it is sufficient to prove 
$\sup_{\bb\in\mathcal{B}}\|\mathbf{I}_n(\bb)\|=o_p(1)$, $\sup_{\bb\in\mathcal{B}}\|\mathbf{II}_n(\bb)\|=o_p(1)$,  and $\sup_{\bb\in\mathcal{B}}\|\mathbf{III}_n(\bb)\|=o_p(1)$.

Firstly, by the boundedness $\bb\in\mathcal{B}$, $Y_i$ and $\bZ_i$, $\|\mathbf{I}_n(\bb)\|\to_p 0$ pointwisely by weak law of large numbers. Since $\mathcal{B}$ is a compact space, we have $\|\mathbf{I}_n(\bb)\|$ uniformly converge to $0$ for $\bb\in\mathcal{B}$. Besides, from the compactness of $\bZ_i$ and $Y_i$, Equation \eqref{SQthm1:e1} implies that 
$\sup_{\bb\in\mathcal{B}}\|\mathbf{II}_n(\bb)\|=o_p(n^{-1/4})$. To prove $\sup_{\bb\in\mathcal{B}}\|\mathbf{III}_n(\bb)\|=o_p(1)$, by the fact that $|I(Y\le u)(Y-u)-I(Y\le v)(Y-v)|\le |v-u|$, we have 
\[
\begin{split}
&|g_{\alpha}(\bZ,\bb,\hat{\bbeta})-g_{\alpha}(\bZ,\bb,\bbeta^\ast)|\\
=&\left|\frac{1}{\alpha}(Y-\bZ^\top\bbeta^\ast)I\{Y\le \bZ^\top\bbeta^\ast\}+\bZ^\top\bbeta^\ast-\frac{1}{\alpha}(Y-\bZ^\top\hat{\bbeta})I(Y\le \bZ^\top\hat{\bbeta})-\bZ^\top\hat{\bbeta}\right|\\
\le&\frac{1}{\alpha}\left|(Y-\bZ^\top\bbeta^\ast)I(Y\le \bZ^\top\bbeta^\ast)-(Y-\bZ^\top\hat{\bbeta})I(Y\le \bZ^\top\hat{\bbeta})\right|+\left|\bZ^\top\bbeta^\ast-\bZ^\top\hat{\bbeta}\right|\\
\le &\left(\frac{1}{\alpha}+1\right)\left|\bZ^\top\bbeta^\ast-\bZ^\top\hat{\bbeta}\right|.
\end{split}
\]

By the assumption that $\|\bZ\|$ is bounded, there exist a constant $C$, such that  $\sup_{\bZ}|\bZ^\top\bbeta^\ast-\bZ^\top\hat{\bbeta}|\le C\|\hat{\bbeta}-{\bbeta}^{\ast}\|$.  Therefore, 
\begin{equation}
\sup_{\bZ,\bb\in\mathcal{B}}|g_{\alpha}(\bZ,\bb,\hat{\bbeta})-g_{\alpha}(\bZ,\bb,\bbeta^\ast)|\le C\left(\frac{1}{\alpha}+1\right)\|\hat{\bbeta}-{\bbeta}^\ast\|.
\label{SQthm1:e2}
\end{equation}

By $\|\hat{\bbeta}-{\bbeta}^\ast\|=o_p(1)$ from Theorem \ref{SQthm0}, we have
\begin{equation}
\begin{split}
\sup_{\bb\in\mathcal{B}}\|\mathbf{III}_n(\bb)\|\le& \sup_i|\kappa_{v,i}|\sup_i\|\bZ_i\|\sup_{\bZ,\bb\in\mathcal{B}}|g_{\alpha}(\bZ,\bb,\hat{\bbeta})-g_{\alpha}(\bZ,\bb,\bbeta^\ast)|\\
\le& \sup_i|\kappa_{v,i}|\sup_i\|\bZ_i\|C\left(\frac{1}{\alpha}+1\right)\|\hat{\bbeta}-{\bbeta}^\ast\|=o_p(1).
\end{split}
\label{SQthm1:e3}
\end{equation}

Therefore, we have $n^{-1}\sum_{i=1}\tilde{\kappa}_{v,i}\bZ_i^\top g_{\alpha,i}(\bb,\hat{\bbeta})$ uniformly converges to $\mathbb{E}\{\kappa_v\bZ g_{\alpha}(\bZ,\bb,\bbeta^\ast)\}$ for $\bb\in\mathcal{B}$. By the definition of $\hat{\bgamma}$ and ${\bgamma}^{\ast}$, we have
\[
\begin{split}
\bzero=&n^{-1}\sum_{i=1}^n\tilde{\kappa}_{v,i}\bZ_i g_{\alpha,i}(\hat{\bgamma},\hat{\bbeta})\\
=&\mathbb{E}\{\kappa_v\bZ g_{\alpha}(\bZ,\hat{\bgamma},\bbeta^\ast)\}-\mathbb{E}\{\kappa_v\bZ g_{\alpha}(\bZ,\bgamma^\ast,\bbeta^\ast)\}+o_p(1)\\
=&\mathbb{E}\{\kappa_v\bZ\bZ^\top\}(\hat{\bgamma}-\bgamma^\ast)+o_p(1).
\end{split}
\]

From Condition (C3), we have $\|\hat{\bgamma}-\bgamma^\ast\|=o_p(1)$.
\end{proof}
\subsection{Proof of Theorem \ref{SQthm2}}

\begin{proof}
In the first step, we will show that $\|\hat{\bgamma}-\bgamma^\ast\|=o_p(n^{-1/4})$ given that $\|\hat{\bbeta}-\bbeta^\ast\|=o_p(n^{-1/4})$. To prove it, we firstly show that $\sup_{\bb}\|n^{-1}\sum_{i=1}^n\tilde{\kappa}_{v,i}\bZ_i^\top g_{\alpha,i}(\bb,\hat{\bbeta})-\mathbb{E}\{\kappa_v\bZ g_{\alpha}(\bZ,\bb,\bbeta^\ast)\}\|=o_p(n^{-1/4})$. From  Equation~\eqref{SQthm1eadd1}, we have 
\[
n^{-1}\sum_{i=1}^n\tilde{\kappa}_{v,i}\bZ_i^\top g_{\alpha,i}(\bb,\hat{\bbeta})-\mathbb{E}\{\kappa_v\bZ g_{\alpha}(\bZ,\bb,\bbeta^\ast)\}=\mathbf{I}_n(\bb)+\mathbf{II}_n(\bb)+\mathbf{III}_n(\bb),
\]
where $\mathbf{I}_n(\bb)$, $\mathbf{II}_n(\bb)$ and $\mathbf{III}_n(\bb)$ are defined in the proof of Theorem \ref{SQthm1}. Thus, we only need to show that $\sup_{\bb\in\mathcal{B}}\|\mathbf{I}_n(\bb)\|=o_p(n^{-1/4})$, $\sup_{\bb\in\mathcal{B}}\|\mathbf{II}_n(\bb)\|=o_p(n^{-1/4})$, and $\sup_{\bb\in\mathcal{B}}\|\mathbf{III}_n(\bb)\|=o_p(n^{-1/4})$. In the proof of Theorem~\ref{SQthm1}, we have shown that $\sup_{\bb\in\mathcal{B}}\|\mathbf{II}_n(\bb)\|=o_p(n^{-1/4})$. From Equations~\eqref{SQthm1:e2} and \eqref{SQthm1:e3} in the proof of Theorem~\ref{SQthm1}, we have
\begin{equation}
\begin{split}
\sup_{\bb\in\mathcal{B}}\|\mathbf{III}_n(\bb)\|\le& \sup_i|\kappa_{v,i}|\sup_i\|\bZ_i\|\sup_{\bb\in\mathcal{B}}|g_{\alpha}(\bZ,\bb,\hat{\bbeta})-g_{\alpha}(\bZ,\bb,\bbeta^\ast)|\\
&\le \sup_i|\kappa_{v,i}|\sup_i\|\bZ_i\|C\left(\frac{1}{\alpha}+1\right)\|\bbeta^\ast-\hat{\bbeta}\|.
\end{split}
\end{equation}
Given that $\|\hat{\bbeta}-\bbeta^\ast\|=o_p(n^{-1/4})$, we have $\sup_{\bb\in\mathcal{B}}\|\mathbf{III}_n(\bb)\|=o_p(n^{-1/4})$. 

It remains to show that $\sup_{\bb\in\mathcal{B}}\|\mathbf{I}_n(\bb)\|=o_p(n^{-1/4})$. To prove it, we firstly show that $\{\kappa_v\bZ g_{\alpha}(\bZ,\bb,\bbeta^\ast):\bb\in\mathcal{B}\}$ is a Donsker class. Given the boundedness of $\bZ$ and $\kappa_{v}$, and the compactness of $\mathcal{B}$, $\left\|\kappa_v\bZ g_{\alpha}(\bZ,\bb,\bbeta^\ast)\right\|$ is uniformly bounded on $\mathcal{B}$. Moreover, for $\bb,\tilde{\bb}\in\mathcal{B}$, given the boundedness of $\bZ$, there exist constant $C_1$, such that 
\[
\|\kappa_v\bZ g_{\alpha}(\bZ,\bb,\bbeta^\ast)-\kappa_v\bZ g_{\alpha}(\bZ,\tilde{\bb},\bbeta^\ast)\|=\|\kappa_v\bZ\bZ^\top(\bb-\tilde{\bb})\|\le C_1\|\bb-\tilde{\bb}\|.
\]
Hence, $\kappa_v\bZ g_{\alpha}(\bZ,\bb,\bbeta^\ast)$ is a type IV (type 4) function defined in \cite{Andrews1994} with $p=2$ and satisfies Ossiander's $L_2$ entropy condition in \cite{Andrews1994}, which implies that  $\{\kappa_v\bZ g_{\alpha}(\bZ,\bb,\bbeta^\ast):\bb\in\mathcal{B}\}$
is a Donsker class. Thus, from the law of iterated logarithm for empirical process on Vapnik-Cervonenkis (VC) class \citep{Alexander1989}, we have $\sup_{\bb\in\mathcal{B}}\|\sum_{i=1}^n\kappa_{v,i}\bZ_i g_{\alpha,i}(\bb,\bbeta^\ast)-\mathbb{E}\{\kappa_v\bZ g_{\alpha}(\bZ,\bb,\bbeta^\ast)\}\|=O(n^{1/2}(\log\log n))$ almost surely, which implies that $\sup_{\bb\in\mathcal{B}}\|\mathbf{I}_n(\bb)\|=o_p(n^{-1/4})$. 
Coupled with Equation~\eqref{SQthm1eadd1}, we have 
\begin{equation}
\sup_{\bb\in\mathcal{B}}\|n^{1/4}[n^{-1}\sum_{i=1}^n\tilde{\kappa}_{v,i}\bZ_i^\top g_{\alpha,i}(\bb,\hat{\bbeta})-\mathbb{E}\{\kappa_v\bZ g_{\alpha}(\bZ,\bb,\bbeta^\ast)\}]\|=o_p(1).
\label{SQthm2e00}
\end{equation}

Since $n^{-1}\sum_{i=1}^n\tilde{\kappa}_{v,i}\bZ_i g_{\alpha,i}(\hat{\bgamma},\hat{\bbeta})=\bzero$ and $\mathbb{E}\{\kappa_v\bZ g_{\alpha}(\bZ,\bgamma^\ast,\bbeta^\ast)\}=\bzero$, we have 
\[
\begin{split}
\bzero=&n^{-1}\sum_{i=1}^n\tilde{\kappa}_{v,i}\bZ_i g_{\alpha,i}(\hat{\bgamma},\hat{\bbeta})\\
=&\mathbb{E}\{\kappa_v\bZ g_{\alpha}(\bZ,\hat{\bgamma},\bbeta^\ast)\}-\mathbb{E}\{\kappa_v\bZ g_{\alpha}(\bZ,\bgamma^\ast,\bbeta^\ast)\}+\Big[n^{-1}\sum_{i=1}^n\tilde{\kappa}_{v,i}\bZ_i^\top g_{\alpha,i}(\hat{\bgamma},\hat{\bbeta})-\mathbb{E}\{\kappa_v\bZ g_{\alpha}(\bZ,\hat{\bgamma},\bbeta^\ast)\}\Big]\\
=&\mathbb{E}\{\kappa_v\bZ\bZ^\top\}(\hat{\bgamma}-\bgamma^\ast)+\Big[n^{-1}\sum_{i=1}^n\tilde{\kappa}_{v,i}\bZ_i^\top g_{\alpha,i}(\hat{\bgamma},\hat{\bbeta})-\mathbb{E}\{\kappa_v\bZ g_{\alpha}(\bZ,\hat{\bgamma},\bbeta^\ast)\}\Big].
\end{split}
\]
From Condition (C3) and Equation~\eqref{SQthm2e00}, we have $\|\hat{\bgamma}-\bgamma^\ast\|=o_p(n^{-1/4})$.

Next step, we will prove that \[
n^{1/2}\mathbb{E}(\kappa_{v}\bZ\bZ^\top)(\hat{\bgamma}-\bgamma^\ast)=-n^{-1/2}\sum_{i=1}^n\tilde{\kappa}_{v,i}\bZ_i g_{\alpha,i}(\bgamma^\ast,\bbeta^\ast)+o_p(1).
\]
Note that 
\begin{equation}
\begin{split}
\bzero=&n^{-1/2}\sum_{i=1}^n\tilde{\kappa}_{v,i}\bZ_i g_{\alpha,i}(\hat{\bgamma},\hat{\bbeta})\\
=&n^{-1/2}\sum_{i=1}^n\tilde{\kappa}_{v,i}\bZ_i g_{\alpha,i}(\bgamma^\ast,\bbeta^\ast)\\
&+n^{-1/2}\sum_{i=1}^n(\tilde{\kappa}_{v,i}-\kappa_{v,i})\bZ_i\{g_{\alpha,i}(\hat{\bgamma},\hat{\bbeta})-g_{\alpha,i}(\bgamma^\ast,\bbeta^\ast)\}\\
&+n^{1/2}\mathbb{E}[\kappa_{v}\bZ \{g_{\alpha}(\bZ,\hat{\bgamma},\hat{\bbeta})-g_{\alpha}(\bZ,\bgamma^\ast,\bbeta^\ast)\}]\\
&+ G_n(\hat{\bgamma},\hat{\bbeta})-G_n(\bgamma^\ast,\bbeta^\ast)\\
=&n^{-1/2}\sum_{i=1}^n\tilde{\kappa}_{v,i}\bZ_i g_{\alpha,i}(\bgamma^\ast,\bbeta^\ast)\\
&+n^{-1/2}\sum_{i=1}^n(\tilde{\kappa}_{v,i}-\kappa_{v,i})\bZ_i \{g_{\alpha,i}(\hat{\bgamma},\hat{\bbeta})-g_{\alpha,i}(\bgamma^\ast,\bbeta^\ast)\}\\
&+n^{1/2}\mathbb{E}[\kappa_{v}\bZ \{g_{\alpha}(\bZ,\bgamma^\ast,\hat{\bbeta})-g_{\alpha}(\bZ,\bgamma^\ast,\bbeta^\ast)\}]\\
&+n^{1/2}\mathbb{E}(\kappa_{v}\bZ\bZ^\top)(\hat{\bgamma}-\bgamma^\ast) \\
&+ G_n(\hat{\bgamma},\hat{\bbeta})-G_n(\bgamma^\ast,\bbeta^\ast),
\end{split}
\label{SQthm2:eq0}
\end{equation}
where $G_n(\bb_1,\bb_2)=n^{-1/2}\sum_{i=1}^n{\kappa}_{v,i}\bZ_i g_{\alpha,i}(\bb_1,\bb_2)-n^{1/2}\mathbb{E}\{\kappa_{v}\bZ g_{\alpha}(\bZ,\bb_1,\bb_2)\}$.

Let   $\mathcal{T}_{n,1}=n^{-1/2}\sum_{i=1}^n(\tilde{\kappa}_{v,i}-\kappa_{v,i})\bZ_i\{g_{\alpha,i}(\hat{\bgamma},\hat{\bbeta})-g_{\alpha,i}(\bgamma^\ast,\bbeta^\ast)\}$, $\mathcal{T}_{n,2}=n^{1/2}\mathbb{E}[\kappa_{v}\bZ \{g_{\alpha}(\bZ,\bgamma^\ast,\hat{\bbeta})-g_{\alpha}(\bZ,\bgamma^\ast,\bbeta^\ast)\}]$, and $\mathcal{T}_{n,3}=G_n(\hat{\bgamma},\hat{\bbeta})-G_n(\bgamma^\ast,\bbeta^\ast)$. In the following, we will show that $\|\mathcal{T}_{n,i}\|=o_p(1)$ for $i=\{1,2,3\}$.

We start with showing $\|\mathcal{T}_{n,1}\|=o_p(1)$. Note that
\begin{equation}
\begin{split}
\|\mathcal{T}_{n,1}\|=&\|n^{-1/2}\sum_{i=1}^n(\tilde{\kappa}_{v,i}-\kappa_{v,i})\bZ_i\{g_{\alpha,i}(\hat{\bgamma},\hat{\bbeta})-g_{\alpha,i}(\bgamma^\ast,\hat{\bbeta})+g_{\alpha,i}(\bgamma^\ast,\hat{\bbeta})-g_{\alpha,i}(\bgamma^\ast,\bbeta^\ast)\}\|\\
\le& \|n^{-1/2}\sum_{i=1}^n(\tilde{\kappa}_{v,i}-\kappa_{v,i})\bZ_i\{g_{\alpha,i}(\hat{\bgamma},\hat{\bbeta})-g_{\alpha,i}(\bgamma^\ast,\hat{\bbeta})\}\|\\
&+\|n^{-1/2}\sum_{i=1}^n(\tilde{\kappa}_{v,i}-\kappa_{v,i})\bZ_i\{g_{\alpha,i}(\bgamma^\ast,\hat{\bbeta})-g_{\alpha,i}(\bgamma^\ast,\bbeta^\ast)\}\|.
\end{split}
\label{SQthm2:eq1_1}
\end{equation}

Moreover, we have
\begin{equation}
\begin{split}
&n^{-1/2}\sum_{i=1}^n(\tilde{\kappa}_{v,i}-\kappa_{v,i})\bZ_i\{g_{\alpha,i}(\hat{\bgamma},\hat{\bbeta})-g_{\alpha,i}(\bgamma^\ast,\hat{\bbeta})\}\\
=&n^{-1/2}\sum_{i=1}^n(\tilde{\kappa}_{v,i}-\kappa_{v,i})\bZ_i\Big[\{\bZ_i^\top\hat{\bgamma}-\frac{1}{\alpha}(Y_i-\bZ_i^\top\hat{\bbeta})I(Y_i\le \bZ_i^\top\hat{\bbeta})-\bZ_i^\top\hat{\bbeta}\}\\
&-\{\bZ_i^\top\bgamma^\ast-\frac{1}{\alpha}(Y_i-\bZ_i^\top\hat{\bbeta})I(Y_i\le \bZ_i^\top\hat{\bbeta})-\bZ_i^\top\hat{\bbeta}\}\Big]\\
=&n^{-1/2}\sum_{i=1}^n(\tilde{\kappa}_{v,i}-\kappa_{v,i})\bZ_i\bZ_i^\top(\hat{\bgamma}-\bgamma^\ast).
\end{split}
\label{SQthm2:eq1_2}
\end{equation}

Also, based on Equation ~\eqref{SQthm1:e2}, we have 
\begin{equation}
\begin{split}
&\|n^{-1/2}\sum_{i=1}^n(\tilde{\kappa}_{v,i}-\kappa_{v,i})\bZ_i\{g_{\alpha,i}(\bgamma^\ast,\hat{\bbeta})-g_{\alpha,i}(\bgamma^\ast,\bbeta^\ast)\}\|\\
\le& n^{-1/2}\sum_{i=1}^n|\tilde{\kappa}_{v,i}-\kappa_{v,i}|\|\bZ_i\| \sup_{\bZ,\bb\in\mathcal{B}}|g_{\alpha}(\bZ,\bb,\hat{\bbeta})-g_{\alpha}(\bZ,\bb,\bbeta^\ast)|\\
\le& n^{-1/2}\sum_{i=1}^n|\tilde{\kappa}_{v,i}-\kappa_{v,i}|\|\bZ_i\| C\left(\frac{1}{\alpha}+1\right)\|\hat{\bbeta}-{\bbeta}^\ast\|.
\end{split}
\label{SQthm2:eq1_3}
\end{equation}

Coupled with the Equations~\eqref{SQthm2:eq1_1}, \eqref{SQthm2:eq1_2} and ~\eqref{SQthm2:eq1_3},  we have
\begin{equation}
\begin{split}
\|\mathcal{T}_{n,1}\|\le& \|n^{-1/2}\sum_{i=1}^n(\tilde{\kappa}_{v,i}-\kappa_{v,i})\bZ_i\bZ_i^\top(\hat{\bgamma}-\bgamma^\ast)\|+n^{-1/2}\sum_{i=1}^n|\tilde{\kappa}_{v,i}-\kappa_{v,i}|\|\bZ_i\|\left(\frac{1}{\alpha}+1\right)C\|\hat{\bbeta}-\bbeta^\ast\|.
\end{split}
\label{SQthm2:eq1}
\end{equation}

Since  $\sup_{i}|\tilde{\kappa}_{v,i}-\kappa_{v,i}|=o_p(n^{-1/4})$, $\|\hat{\bbeta}-\bbeta^\ast\|=o_p(n^{-1/4})$, and
$\|\hat{\bgamma}-\bgamma^\ast\|=o_p(n^{-1/4})$, \eqref{SQthm2:eq1} implies that $\|\mathcal{T}_{n,1}\|=o_p(1)$.

Next, we will show that  $\|\mathcal{T}_{n,2}\|=o_p(1)$. Let $\bigtriangledown$ denotes vector differential operator. By a Taylor's expansion, we have
\begin{equation}
\begin{split}
&\|n^{1/2}\mathbb{E}[\kappa_{v}\bZ \{g_{\alpha}(\bZ,\bgamma^\ast,\hat{\bbeta})-g_{\alpha}(\bZ,\bgamma^\ast,\bbeta^\ast)\}]\|\\
&=\|n^{1/2}\mathbb{E}\{\kappa_{v}\bZ\bigtriangledown_{\bv}g_{\alpha}(\bZ,\bgamma^\ast,\bv)|_{\bv=\bbeta^\ast}(\hat{\bbeta}-\bbeta^\ast)\}\|+n^{1/2}O_p(\|\hat{\bbeta}-\bbeta^\ast\|^2).
\end{split}
\label{SQthm2:eq2}
\end{equation}

From the definition of $\bbeta^\ast$, we have
\[
\begin{split}
&\mathbb{E}\{\kappa_{v}\bZ\bigtriangledown_{\bv}g_{\alpha}(\bZ,\bgamma^\ast,\bv)|_{\bv=\bbeta^\ast}\}\\
=&\Pr(D_1>D_0)\mathbb{E}[\bZ\bigtriangledown_{\bv}g_{\alpha}(\bZ,\bgamma^\ast,\bv)|_{\bv=\bbeta^\ast}|D_1>D_0]\\
=&\Pr(D_1>D_0)\mathbb{E}\left[\bZ\left\{\frac{1}{\alpha}F_{Y}(\bZ^\top\bbeta^\ast|\bZ,D_1>D_0)-1\right\}\bZ^\top|D_1>D_0\right]\\
=&\bzero.
\end{split}
\]

Thus, by~\eqref{SQthm2:eq2}, we have $\|\mathcal{T}_{n,2}\|\le n^{1/2}O_p(\|\hat{\bbeta}-\bbeta^\ast\|^2)=o_p(1)$ by the assumption that $\|\hat{\bbeta}-\bbeta^\ast\|=o_p(n^{-1/4})$.

To prove $\|\mathcal{T}_{n,3}\|=o_p(1)$, we adopt a similar argument in the proof of Lemma 3 in \cite{Barendse2020}. Specifically, 
for any $\bb_1\in\mathcal{B}$ and $\bb_2\in\mathcal{K}$, $\|\kappa_{v}\bZ g_{\alpha}(\bZ,\bb_1,\bb_2)\|$ is uniformly bounded over $\mathcal{B}\times\mathcal{K}$ based on the boundedness of $\kappa_v$ and $\bZ$, and the compactness of $\mathcal{B}$ and $\mathcal{K}$. Moreover, for any $\bb_1,\tilde{\bb}_1\in\mathcal{B}$ and $\bb_2,\tilde{\bb}_2\in\mathcal{K}$, we have 
\begin{equation}
\begin{split}
&\|\kappa_{v}\bZ g_{\alpha}(\bZ,\bb_1,\bb_2)-\kappa_{v}\bZ g_{\alpha}(\bZ,\tilde{\bb}_1,\tilde{\bb}_2)\|\\
\le&\|\kappa_{v}\bZ g_{\alpha}(\bZ,\bb_1,\bb_2)-\kappa_{v}\bZ g_{\alpha}(\bZ,\tilde{\bb}_1,\bb_2)\|+\|\kappa_{v}\bZ g_{\alpha}(\bZ,\tilde{\bb}_1,\bb_2)-\kappa_{v}\bZ g_{\alpha}(\bZ,\tilde{\bb}_1,\tilde{\bb}_2)\|\\
\le& \sup\|\bZ\|\|\bb_1-\tilde{\bb}_1\|+\sup\|\bZ\| \left(\frac{1}{\alpha}+1\right)C\|\bb_2-\tilde{\bb}_2\|.
\end{split}
\label{SQthm2:eq2_1}
\end{equation}
The last equation follows from Equation \eqref{SQthm1:e2}. Hence, $\kappa_{v}\bZ g_{\alpha}(\bZ,\bb_1,\bb_2)$ is type IV (type 4) function in \cite{Andrews1994} and satisfies Ossiander's $L_2$ entropy condition. Following the similar argument in the proof of Lemma 3 in \cite{Barendse2020}, \cite{Doukhan1995} and Equation~\eqref{SQthm2:eq2_1} imply that stochastic equicontinuity holds for $\kappa_{v}\bZ g_{\alpha}(\bZ,\bb_1,\bb_2)$, which implies $G_n(\hat{\bgamma},\hat{\bbeta})-G_n(\bgamma^\ast,\bbeta^\ast)=o_p(1)$.

From equation \eqref{SQthm2:eq0}, we have the following equation 
\[
n^{1/2}\mathbb{E}(\kappa_{v}\bZ\bZ^\top)(\hat{\bgamma}-\bgamma^\ast)=-n^{-1/2}\sum_{i=1}^n\tilde{\kappa}_{v,i}\bZ_i g_{\alpha,i}(\bgamma^\ast,\bbeta^\ast)+o_p(1).
\]

Coupled with Lemma \ref{SQlemma1}, we have 
\begin{equation}
n^{1/2}\mathbb{E}\{I(D_1>D_0)\bZ\bZ^\top\}(\hat{\bgamma}-\bgamma^\ast)=n^{-1/2}\sum_{i=1}^n\bPsi_i(\alpha)+o_p(1).
\label{SQthm2:eq4}
\end{equation}
Therefore, $n^{1/2}(\hat{\bgamma}-\bgamma^\ast)\to_d N(0,\bJ_2^{-1}\bOmega_2\bJ_2^{-1})$, where $\bJ_2=\mathbb{E}\{I(D_1>D_0)\bZ\bZ^\top\}$ and $\bOmega_2=\mathbb{E}\{\bPsi(\alpha)\bPsi(\alpha)^\top\}$.
\end{proof}

\section{Theoretical Justification for the Standard Nonparametric Bootstrap Inference
Procedure}

 In this section,  we will prove that $n^{1/2}\{\bbeta^{\dagger}-\hat{\bbeta}\}$ given the observed data is asymptotically equivalent to $n^{1/2}\{\hat{\bbeta}-\bbeta^\ast\}$ given Conditions (C1)-(C8), where $\bbeta^{\dagger}$ is bootstrapped counterpart of $\hat{\bbeta}$. We will also prove that $n^{1/2}\{\bgamma^{\dagger}-\hat{\bgamma}\}$ given the observed data is asymptotically equivalent to $n^{1/2}\{\hat{\bgamma}-\bgamma^\ast\}$ given Conditions (C1)-(C8), where $\bgamma^{\dagger}$ is bootstrapped counterpart of $\hat{\bgamma}$.

We firstly state  data notations that used in this justification. Let $\bU=\{D,\bX,V,Y\}$, and $\bU_i=\{D_i,\bX_i,V_i,Y_i\}$. Denote $\delta_{\bU_i}$ as the probability measure that assigns a mass of 1 to $\bU_i$. The empirical
measure based on the observed data is given by $\mathbb{P}_n=n^{-1}\sum_{i=1}^n\delta_{\bU_i}$. The bootstrap empirical
measure corresponding to the standard nonparametric bootstrap inference procedure
is given by $\hat{\mathbb{P}}_n=n^{-1}\sum_{i=1}^n W_{n,i}\delta_{\bU_i}$, where $\vec{\bW}_n= (W_{n1},\dots,W_{nn})^{\top}$ is a multinomial
vector with probabilities $(1/n,\dots,1/n)^{\top}$ and index $n$, and $\vec{\bW}_n$ is independent of the observed data $\{\bU_i\}_{i=1}^n$. Let $\mathbb{G}_n=n^{-1/2}\sum_{i=1}^n(\delta_{\bU_i}-\mathbb{P})$, and $\hat{\mathbb{G}}_n=n^{-1/2}\sum_{i=1}^nW_{n,i}(\delta_{\bU_i}-\mathbb{P}_n)$, where $\mathbb{P}$ is the probability measure that governs $\bU$.


Throughout this section, for notational convenience, we include a superscript for the bootstrapped counterparts of the estimators, estimating functions, and other quantities of interest.  
Specifically, let
\[
\pi^{\dagger}(\bx)=\frac{\sum_{i=1}^nW_{n,i}\mathcal{K}^{\ast}_{\sigma_1}(\bx-\bX_i)V_i}{\sum_{i=1}^nW_{n,i}\mathcal{K}^{\ast}_{\sigma_1}(\bx-\bX_i)},\ \
v_{d}^{\dagger}(y,\bx)=\frac{\sum_{i=1}^n W_{n,i}I(D_i=d)\mathcal{K}^{\ast\ast}_{\sigma_2}\{(y,\bx^\top)^\top-(Y_i,\bX_i^\top)^\top\}V_i}{\sum_{i=1}^n W_{n,i}I(D_i=d)\mathcal{K}^{\ast\ast}_{\sigma_2}\{(y,\bx^\top)^\top-(Y_i,\bX_i^\top)^\top\}},
\]
$v^{\dagger}(Y_i,\bZ_i)=I(D_i=1)v_{1}^{\dagger}(Y_i,\bX_i)+I(D_i=0)v_{0}^{\dagger}(Y_i,\bX_i)$,  $\kappa_{v,i}^{\dagger}=1-\frac{D_i\{1- v^{\dagger}(Y_i,\bZ_i)\}}{1-\pi^{\dagger}(\bX_i)}-\frac{(1-D_i)v^{\dagger}(Y_i,\bZ_i)}{\pi^{\dagger}(\bX_i)}$, and $\tilde{\kappa}_{v,i}^{\dagger}=\min\{\max(\kappa_{v,i}^{\dagger},c_{l,n}),c_{u,n}\}$. Let $\pi_i^{\dagger}$ and $v_i^{\dagger}$ be shorthand notation for $\pi^{\dagger}(\bX_i)$ and
$v^{\dagger}(Y_i,\bZ_i)$ respectively. 

Based on the results on the error rate of empirical bootstrap approximation rate for kernel estimator \citep{Neumann1998,Chernozhukov2014,Chernozhukov2016}, under Condition (C8), we have $\sup_i|\pi_i^{\dagger}-\pi_i|=o_p(n^{-1/4})$, $\sup_i|v_i^{\dagger}-v_i|=o_p(n^{-1/4})$ and $\sup_i|\kappa_{v,i}^{\dagger}-\kappa_{v,i}|=o_p(n^{-1/4})$. Under Condition (C6), we have $\sup_i|\tilde{\kappa}_{v,i}^{\dagger}-\kappa_{v,i}|=o_p(n^{-1/4})$.

We start with some lemmas that will be helpful in justifying the standard non-parametric bootstrap inference procedure. The proofs of the Lemmas are deferred to  ~\ref{appendix:C}.

\begin{lemma}
Under Conditions (C1)--(C8), we have
\[
n^{-1/2}\sum_{i=1}^nW_{n,i}\tilde{\kappa}_{v,i}^{\dagger}\bZ_ig_{\alpha,i}(\bgamma^\ast,\bbeta^\ast)=n^{-1/2}\sum_{i=1}^nW_{n,i}\bPsi_i(\alpha)+o_p(1),
\]
where $\bPsi_i(\alpha)=\bfm_2(Y_i,\bZ_i,\alpha)\{1-\frac{D_i(1-V_i)}{1-\pi_i}-\frac{(1-D_i)V_i}{\pi_i}\}+\bH_2(\bX_i,\alpha)\{V_i-\pi(\bX_i)\}$, with
$\bfm_2(Y_i,\bZ_i,\alpha)=\bZ_ig_{\alpha,i}(\bgamma^\ast,\bbeta^\ast)$ and $\bH_2(\bX_i,\alpha)=\mathbb{E}[\bfm_2(Y_i,\bZ_i,\alpha)\{\frac{(1-D_i)v_i}{\pi_i^2}-\frac{D_i(1-v_i)}{(1-\pi_i)^2}\}|\bX_i]$.
\label{SQlemma3}
\end{lemma}

\begin{lemma}
Under Conditions (C1)--(C8), we have
\[
n^{-1/2}\sum_{i=1}^nW_{n,i}\tilde{\kappa}_{v,i}^{\dagger}\bZ_i\{\alpha-I(Y_i<\bZ_i^\top\bbeta^\ast)\}=n^{-1/2}\sum_{i=1}^nW_{n,i}\bPhi_i(\alpha)+o_p(1),
\]
where $\bPhi_i(\alpha)=\bfm_1(Y_i,\bZ_i,\alpha)\{1-\frac{D_i(1-V_i)}{1-\pi_i}-\frac{(1-D_i)V_i}{\pi_i}\}+\bH_1(\bX_i,\alpha)\{V_i-\pi(\bX_i)\}$. Here 
$\bfm_1(Y_i,\bZ_i,\alpha)=\bZ_i\{\alpha-I(Y_i<\bZ_i^\top\bbeta^\ast)\}$ and $\bH_1(\bX_i,\alpha)=\mathbb{E}[\bfm_1(Y_i,\bZ_i,\alpha)\{\frac{(1-D_i)v_i}{\pi_i^2}-\frac{D_i(1-v_i)}{(1-\pi_i)^2}\}|\bX_i]$.
\label{SQlemma4}
\end{lemma}

\subsection{Theoretical Justification for the Standard Nonparametric Bootstrap Inference
Procedure for $\bbeta^\ast$ \label{bootbeta}}

We will prove that  the conditional distribution of  $n^{1/2}\{\bbeta^{\dagger}-\hat{\bbeta}\}$ given the observed data is asymptotically equivalent to $n^{1/2}\{\hat{\bbeta}-\bbeta^\ast\}$ to justify  the standard non-parametric bootstrap inference procedure for $\bbeta^\ast$. To prove this result, We only need to prove that the conditional distribution of  $n^{1/2}\{\bbeta^{\dagger}-\hat{\bbeta}\}$ given the observed data is asymptotically equivalent to $N(\bzero, \bJ_1^{-1}\bOmega_1\bJ_1^{-1})$. This result can be proved by  adapting the similar arguments in the proof of Theorem \ref{SQthm0} and the techniques in the justification of the standard nonparametric bootstrap inference procedure in \cite{Wei2021}. In detail, let
 $L_n^{\dagger}(\bzeta,\kappa)=\sum_{i=1}^n W_{n,i}l_i(\bzeta,\kappa)$ with
\[
l_i(\bzeta,\kappa)=\kappa_{v,i}\{\rho_{\alpha}(\epsilon_{i}-n^{-1/2}\bZ_i^\top\bzeta)-\rho_{\alpha}(\epsilon_{i})\},
\]
where $\epsilon_i=Y_i-\bZ_i^\top\bbeta^\ast$ and $\rho_{\alpha}(u)=u\{\alpha-I(u\le 0)\}$.  Similarly, let  $L_n^{\dagger}(\bzeta,\tilde{\kappa}^{\dagger})=\sum_{i=1}^n W_{n,i}l_i(\bzeta,\tilde{\kappa}^{\dagger})$ with
$l_i(\bzeta,\tilde{\kappa}^{\dagger})=\tilde{\kappa}_{v,i}^{\dagger}\{\rho_{\alpha}(\epsilon_{i}-n^{-1/2}\bZ_i^\top\bzeta)-\rho_{\alpha}(\epsilon_{i})\}.$

The function $L_n(\bzeta,\tilde{\kappa})$ is a convex function with respect to $\bzeta$ and a minimizer takes the form $\bzeta_n^{\dagger}=n^{1/2}(\bbeta^{\dagger}-\bbeta^\ast)$. 

Note that  $\vec{\bW}_n$ is independent of the observed data, following the similar argument in the proof of Equation \eqref{SQthm0e0} in Theorem \ref{SQthm0}, by a Taylor's  expansion, we have
\begin{equation}
\mathbb{E}\{L_n^{\dagger}(\bzeta,\kappa)\}=\frac{1}{2}\bzeta^\top\bJ_1\bzeta+o_p(1).
\label{bootthm0e0}
\end{equation}

Similarly, by adapting the similar argument in the proof of Equation \eqref{SQthm0e0} in Theorem \ref{SQthm0}, we can obtain 
\begin{equation}
L_n^{\dagger}(\bzeta,\tilde{\kappa}^\dagger)=\mathbb{E}\{L_n^{\dagger}(\bzeta,\kappa)\}-n^{-1/2}\sum_{i=1}^nW_{n,i}\tilde{\kappa}_{v,i}^{\dagger}\bZ_i^\top\bzeta\{\alpha-I(\epsilon_i\le 0)\}+o_p(1).
\label{bootthm0e1}
\end{equation}

From Equations~\eqref{bootthm0e0} and \eqref{bootthm0e1}, for a given $\bzeta$,
\[
\begin{split}
L_n^\dagger(\bzeta,\tilde{\kappa}^{\dagger})=&\mathbb{E}\{L_n^{\dagger}(\bzeta,\kappa)\}-n^{-1/2}\sum_{i=1}^nW_{n,i}\tilde{\kappa}_{v,i}^{\dagger}\bZ_i^\top\bzeta\{\alpha-I(\epsilon_i\le 0)\}+o_p(1)\\
=&\frac{1}{2}\bzeta^\top\bJ_1\bzeta-n^{-1/2}\sum_{i=1}^nW_{n,i}\tilde{\kappa}_{v,i}^{\dagger}\bZ_i^\top\bzeta\{\alpha-I(\epsilon_i\le 0)\}+o_p(1).
\end{split}
\]
Since $L_n^{\dagger}(\bzeta,\tilde{\kappa}^{\dagger})+n^{-1/2}\sum_{i=1}^nW_{n,i}\tilde{\kappa}_{v,i}^{\dagger}\bZ_i^\top\bzeta\{\alpha-I(\epsilon_i\le 0)\}$ is convex in $\bzeta$. From Pollard's convexity lemma \citep{Pollard1991}, for any compact subset $\mathcal{T}\subset\mathbb{R}^{l+1}$,  we have
\begin{equation}
\sup_{\bzeta\in\mathcal{T}}|L_n^{\dagger}(\bzeta,\tilde{\kappa}^{\dagger})+n^{-1/2}\sum_{i=1}^nW_{n,i}\tilde{\kappa}_{v,i}^{\dagger}\bZ_i^\top\bzeta\{\alpha-I(\epsilon_i\le 0)\}-\frac{1}{2}\bzeta^\top\bJ_1\bzeta|=o_p(1).
\label{bootthm0e5}
\end{equation}
Let $\bfeta_n^{\dagger}=\bJ_1^{-1}n^{-1/2}\sum_{i=1}^nW_{n,i}\tilde{\kappa}_{v,i}^{\dagger}\bZ_i\{\alpha-I(\epsilon_i\le 0)\}$. Note that
\[
\frac{1}{2}(\bzeta-\bfeta_n^{\dagger})^\top\bJ_1(\bzeta-\bfeta_n^{\dagger})=\frac{1}{2}\bzeta^\top\bJ_1\bzeta-n^{-1/2}\sum_{i=1}^nW_{n,i}\tilde{\kappa}_{v,i}^{\dagger}\bZ_i^\top\bzeta\{\alpha-I(\epsilon_i\le 0)\}+\frac{1}{2}(\bfeta_n^{\dagger})^\top\bJ_1\bfeta_n^{\dagger}.
\]
From Equation~\eqref{bootthm0e5}, for any compact subset $\mathcal{T}\subset\mathbb{R}^{l+1}$, we have
\[
\sup_{\bzeta\in\mathcal{T}}|L_n^{\dagger}(\bzeta,\tilde{\kappa}^{\dagger})-\frac{1}{2}(\bzeta-\bfeta_n^{\dagger})^\top\bJ_1(\bzeta-\bfeta_n^{\dagger})+\frac{1}{2}(\bfeta_n^{\dagger})^\top\bJ_1\bfeta_n^{\dagger}|=o_p(1).
\]

Finally, by an application of  Lemma 3 in \cite{Buchinsky1998}, we have $\bzeta_n^{\dagger}=\bfeta_n^{\dagger}+o_p(1)$. From Lemma \ref{SQlemma4},  we have 
\begin{equation}
n^{1/2}(\bbeta^{\dagger}-\bbeta^\ast)=n^{-1/2}\sum_{i=1}^nW_{n,i}\bPhi_i(\alpha)+o_p(1).
\label{bootthm0e6}
\end{equation}

Note that $n^{1/2}(\hat{\bbeta}-\bbeta^\ast)=n^{-1/2}\sum_{i=1}^n\bPhi_i(\alpha)+o_p(1)$ from the proof of Theorem \ref{SQthm0}. Coupled with Equation \eqref{bootthm0e6} we have 
\begin{equation}
n^{1/2}(\bbeta^{\dagger}-\hat{\bbeta})=n^{-1/2}\sum_{i=1}^nW_{n,i}\bPhi_i(\alpha)-n^{-1/2}\sum_{i=1}^n\bPhi_i(\alpha)+o_p(1).
\label{bootthm0e7}
\end{equation}

To prove that $n^{1/2}\{\bbeta^{\dagger}-\hat{\bbeta}\}$ given the observed data is asymptotically equivalent to $N(\bzero, \bJ_1^{-1}\bOmega_1\bJ_1^{-1})$,  we can adapt similar arguments in \cite{Wei2021}. In detail, since $\{\bPhi_i(\alpha):\alpha\in(0,1)\}$ is a Donsker class from the definition of $\bPhi_i(\alpha)$, it is implied by bootstrap consistency for Donsker classes (e.g. Theorem 2.6 of \cite{Kosorok2008}) that the difference between the conditional random law of $\hat{\mathbb{G}}_n\bPhi_i(\alpha)$ and the unconditional law of $\mathbb{G}_n\bPhi_i(\alpha)$ converges to zero almost surely. It implies that the conditional distribution of $n^{-1/2}\sum_{i=1}^nW_{n,i}\bPhi_i(\alpha)-n^{-1/2}\sum_{i=1}^n\bPhi_i(\alpha)$ given the observed data is asymptotically equivalent to the distribution of $n^{-1/2}\sum_{i=1}^n\bPhi_i(\alpha)-n^{1/2}\mathbb{E}\{\bPhi_i(\alpha)\}$. Note that $\mathbb{E}\{\bPhi_i(\alpha)\}=0$ and $\|n^{-1/2}\sum_{i=1}^n\bPhi_i(\alpha)\|=O_p(1)$, and the distribution of $n^{-1/2}\sum_{i=1}^n\bPhi_i(\alpha)-n^{1/2}\mathbb{E}\{\bPhi_i(\alpha)\}$ is asymptotically equivalent to $N(\bzero, \bJ_1^{-1}\bOmega_1\bJ_1^{-1})$ given the result of  Theorem \ref{SQthm0}. We then obtained that $n^{1/2}\{\bbeta^{\dagger}-\hat{\bbeta}\}$ given the observed data is asymptotically equivalent to $N(\bzero, \bJ_1^{-1}\bOmega_1\bJ_1^{-1})$.

\subsection{Theoretical Justification for the Standard Nonparametric Bootstrap Inference
Procedure for $\bgamma^\ast$}

We will prove that  the conditional distribution of  $n^{1/2}\{\bgamma^{\dagger}-\hat{\bgamma}\}$ given the observed data is asymptotically
equivalent to $N(\bzero, \bJ_2^{-1}\bOmega_2\bJ_2^{-1})$ to justify the  standard nonparametric bootstrap inference
procedure for $\gamma^\ast$. Its proof can be given by adapting the similar techniques in  \ref{bootbeta} and Theorem \ref{SQthm2}. 

Firstly,  we will show that 
\begin{equation}
\sup_{\bb\in\mathcal{B}}\|n^{-1}\sum_{i=1}^nW_{n,i}\tilde{\kappa}_{v,i}^{\dagger}\bZ_i^\top g_{\alpha,i}(\bb,\bbeta^{\dagger})-\mathbb{E}\{\kappa_v\bZ g_{\alpha}(\bZ,\bb,\bbeta^\ast)\}\|=o_p(n^{-1/4})
\label{bootthm2e0}
\end{equation}. 

Adapting the similar argument in Equation~\eqref{SQthm1eadd1}, we have 
\begin{equation}
\begin{split}
&n^{-1}\sum_{i=1}^nW_{n,i}\tilde{\kappa}_{v,i}^{\dagger}\bZ_i g_{\alpha,i}(\bb,\bbeta^{\dagger})\\
=&\mathbb{E}\{\kappa_v\bZ g_{\alpha}(\bZ,\bb,\bbeta^\ast)\}+n^{-1}\sum_{i=1}^nW_{n,i}{\kappa}_{v,i}\bZ_i g_{\alpha,i}(\bb,\bbeta^\ast)-\mathbb{E}\{\kappa_v\bZ g_{\alpha}(\bZ,\bb,\bbeta^\ast)\}\\
&+n^{-1}\sum_{i=1}^nW_{n,i}(\tilde{\kappa}_{v,i}^{\dagger}-\kappa_{v,i})\bZ_i g_{\alpha,i}(\bb,\bbeta^{\dagger})+n^{-1}\sum_{i=1}^nW_{n,i}{\kappa}_{v,i}\bZ_i \{g_{\alpha,i}(\bb,\bbeta^{\dagger})-g_{\alpha,i}(\bb,\bbeta^\ast)\}.\\
\end{split}
\label{bootthm1eadd1}
\end{equation}

Let $\mathbf{I}_n^{\dagger}(\bb)=n^{-1}\sum_{i=1}^nW_{n,i}{\kappa}_{v,i}\bZ_i g_{\alpha,i}(\bb,\bbeta^\ast)-\mathbb{E}\{\kappa_v\bZ g_{\alpha}(\bZ,\bb,\bbeta^\ast)\}$, $\mathbf{II}_n^{\dagger}(\bb)=n^{-1}\sum_{i=1}^nW_{n,i}(\tilde{\kappa}_{v,i}^{\dagger}-\kappa_{v,i})\bZ_i\cdot g_{\alpha,i}(\bb,\bbeta^{\dagger})$, and $\mathbf{III}_n^{\dagger}(\bb)=n^{-1}\sum_{i=1}^nW_{n,i}{\kappa}_{v,i}\bZ_i \{g_{\alpha,i}(\bb,\bbeta^{\dagger})-g_{\alpha,i}(\bb,\bbeta^\ast)\}$.  To prove Equation~\eqref{bootthm2e0}, it is sufficient to prove 
$\sup_{\bb\in\mathcal{B}}\|\mathbf{I}_n^{\dagger}(\bb)\|=o_p(n^{-1/4})$, $\sup_{\bb\in\mathcal{B}}\|\mathbf{II}_n^{\dagger}(\bb)\|=o_p(n^{-1/4})$,  and $\sup_{\bb\in\mathcal{B}}\|\mathbf{III}_n^{\dagger}(\bb)\|=o_p(n^{-1/4})$. 

Note that $\sup_{i}|\tilde{\kappa}_{v,i}^{\dagger}-\kappa_{v,i}|=o_p(n^{-1/4})$. Adapting the similar arguments in Theorem~\ref{SQthm1}, we can shown that $\sup_{\bb\in\mathcal{B}}\|\mathbf{II}_n^{\dagger}(\bb)\|=o_p(n^{-1/4})$. Following the similar arguments in Equations~\eqref{SQthm1:e2} and \eqref{SQthm1:e3} in the proof of Theorem~\ref{SQthm1}, we have
\begin{equation}
\begin{split}
\sup_{\bb\in\mathcal{B}}\|\mathbf{III}_n^{\dagger}(\bb)\|\le&\Big(\frac{1}{n}\sum_{i=1}^nW_{n,i} \Big)\sup_i|\kappa_{v,i}|\sup_i\|\bZ_i\|\sup_{\bb\in\mathcal{B}}|g_{\alpha}(\bZ,\bb,\bbeta^{\dagger})-g_{\alpha}(\bZ,\bb,\bbeta^\ast)|\\
&\le \sup_i|\kappa_{v,i}|\sup_i\|\bZ_i\|C\left(\frac{1}{\alpha}+1\right)\|\bbeta^\ast-\bbeta^{\dagger}\|.
\end{split}
\end{equation}
Given that $\|\bbeta^{\dagger}-\bbeta^\ast\|=o_p(n^{-1/4})$ from the result in \ref{bootbeta}, we have $\sup_{\bb\in\mathcal{B}}\|\mathbf{III}_n^{\dagger}(\bb)\|=o_p(n^{-1/4})$. 

To prove $\sup_{\bb\in\mathcal{B}}\|\mathbf{I}_n^{\dagger}(\bb)\|=o_p(n^{-1/4})$, we first know that  $\{\kappa_v\bZ g_{\alpha}(\bZ,\bb,\bbeta^\ast):\bb\in\mathcal{B}\}$ is a Donsker class from the proof in Theorem~\eqref{SQthm2}. Adapting the similar arguments in Theorem~\eqref{SQthm2}, we have
\begin{equation}
\sup_{b\in\mathcal{B}}\|n^{-1}\sum_{i=1}^n{\kappa}_{v,i}\bZ_i g_{\alpha,i}(\bb,\bbeta^\ast)-\mathbb{E}\{\kappa_v\bZ g_{\alpha}(\bZ,\bb,\bbeta^\ast)\}\|=o_p(n^{-1/4}).
\label{bootthm2eq01}
\end{equation}
By bootstrap consistency for Donsker class (e.g Theorem 2.6 in \cite{Kosorok2008}), we have \\$n^{1/2}\{n^{-1}\sum_{i=1}^nW_{n,i}{\kappa}_{v,i}\bZ_i g_{\alpha,i}(\bb,\bbeta^\ast)-n^{-1}\sum_{i=1}^n{\kappa}_{v,i}\bZ_i g_{\alpha,i}(\bb,\bbeta^\ast)\}$ converges weakly to a tight mean zero Gaussian process. It implies that 
\begin{equation}
\sup_{b\in\mathcal{B}}\|n^{-1}\sum_{i=1}^nW_{n,i}{\kappa}_{v,i}\bZ_i g_{\alpha,i}(\bb,\bbeta^\ast)-n^{-1}\sum_{i=1}^n{\kappa}_{v,i}\bZ_i g_{\alpha,i}(\bb,\bbeta^\ast)\|=o_p(n^{-1/4}).
\label{bootthm2eq02}
\end{equation}
Based on Equations~\eqref{bootthm2eq01} and \eqref{bootthm2eq02}, we can show that $\sup_{\bb\in\mathcal{B}}\|\mathbf{I}_n^{\dagger}(\bb)\|=o_p(n^{-1/4})$.

Since $n^{-1}\sum_{i=1}^nW_{n,i}\tilde{\kappa}_{v,i}^{\dagger}\bZ_i^\top g_{\alpha,i}(\bgamma^{\dagger},\bbeta^{\dagger})=\bzero$ and $\mathbb{E}\{\kappa_v\bZ g_{\alpha}(\bZ,\bgamma^\ast,\bbeta^\ast)\}=\bzero$, we have 
\[
\begin{split}
\bzero=&n^{-1}\sum_{i=1}^nW_{n,i}\tilde{\kappa}_{v,i}^{\dagger}\bZ_i g_{\alpha,i}(\bgamma^{\dagger},\bbeta^{\dagger})\\
=&\mathbb{E}\{\kappa_v\bZ g_{\alpha}(\bZ,\bgamma^\dagger,\bbeta^\ast)\}-\mathbb{E}\{\kappa_v\bZ g_{\alpha}(\bZ,\bgamma^\ast,\bbeta^\ast)\}+\Big[n^{-1}\sum_{i=1}^nW_{n,i}\tilde{\kappa}_{v,i}^{\dagger}\bZ_i^\top g_{\alpha,i}(\bgamma^{\dagger},\bbeta^{\dagger})-\mathbb{E}\{\kappa_v\bZ g_{\alpha}(\bZ,\bgamma^{\dagger},\bbeta^{\dagger})\}\Big]\\
=&\mathbb{E}\{\kappa_v\bZ\bZ^\top\}(\bgamma^{\dagger}-\bgamma^\ast)+\Big[n^{-1}\sum_{i=1}^nW_{n,i}\tilde{\kappa}_{v,i}\bZ_i^\top g_{\alpha,i}(\bgamma^{\dagger},\bbeta^{\dagger})-\mathbb{E}\{\kappa_v\bZ g_{\alpha}(\bZ,\bgamma^{\dagger},\bbeta^\ast)\}\Big].
\end{split}
\]
From Condition (C3) and Equation~\eqref{bootthm2e0}, we have $\|\hat{\bgamma}-\bgamma^\ast\|=o_p(n^{-1/4})$.

Next step, we will prove that \[
n^{1/2}\mathbb{E}(\kappa_{v}\bZ\bZ^\top)(\bgamma^{\dagger}-\bgamma^\ast)=-n^{-1/2}\sum_{i=1}^nW_{n,i}\tilde{\kappa}_{v,i}^{\dagger}\bZ_i g_{\alpha,i}(\bgamma^\ast,\bbeta^\ast)+o_p(1).
\]
Note that 
\begin{equation}
\begin{split}
\bzero=&n^{-1/2}\sum_{i=1}^nW_{n,i}\tilde{\kappa}_{v,i}^{\dagger}\bZ_i g_{\alpha,i}(\bgamma^{\dagger},\bbeta^{\dagger})\\
=&n^{-1/2}\sum_{i=1}^nW_{n,i}\tilde{\kappa}_{v,i}^{\dagger}\bZ_i g_{\alpha,i}(\bgamma^\ast,\bbeta^\ast)\\
&+n^{-1/2}\sum_{i=1}^nW_{n,i}(\tilde{\kappa}_{v,i}^{\dagger}-\kappa_{v,i})\bZ_i\{g_{\alpha,i}(\bgamma^{\dagger},\bbeta^{\dagger})-g_{\alpha,i}(\bgamma^\ast,\bbeta^\ast)\}\\
&+n^{1/2}\mathbb{E}[\kappa_{v}\bZ \{g_{\alpha}(\bZ,\bgamma^{\dagger},\bbeta^{\dagger})-g_{\alpha}(\bZ,\bgamma^\ast,\bbeta^\ast)\}]\\
&+ G_n^{\dagger}(\bgamma^{\dagger},\bbeta^{\dagger})-G_n^{\dagger}(\bgamma^\ast,\bbeta^\ast)\\
=&n^{-1/2}\sum_{i=1}^n W_{n,i}\tilde{\kappa}_{v,i}^{\dagger}\bZ_i g_{\alpha,i}(\bgamma^\ast,\bbeta^\ast)\\
&+n^{-1/2}\sum_{i=1}^nW_{n,i}(\tilde{\kappa}_{v,i}^{\dagger}-\kappa_{v,i})\bZ_i \{g_{\alpha,i}(\bgamma^{\dagger},\bbeta^{\dagger})-g_{\alpha,i}(\bgamma^\ast,\bbeta^\ast)\}\\
&+n^{1/2}\mathbb{E}[\kappa_{v}\bZ \{g_{\alpha}(\bZ,\bgamma^\ast,\bbeta^{\dagger})-g_{\alpha}(\bZ,\bgamma^\ast,\bbeta^\ast)\}]\\
&+n^{1/2}\mathbb{E}(\kappa_{v}\bZ\bZ^\top)(\bgamma^{\dagger}-\bgamma^\ast) \\
&+ G_n^{\dagger}(\bgamma^{\dagger},\bbeta^{\dagger})-G_n^{\dagger}(\bgamma^\ast,\bbeta^\ast),
\end{split}
\label{bootthm2:eq0}
\end{equation}
where $G_n^{\dagger}(\bb_1,\bb_2)=n^{-1/2}\sum_{i=1}^nW_{n,i}{\kappa}_{v,i}\bZ_i g_{\alpha,i}(\bb_1,\bb_2)-n^{1/2}\mathbb{E}\{\kappa_{v}\bZ g_{\alpha}(\bZ,\bb_1,\bb_2)\}$.

Let   $\mathcal{T}_{n,1}^{\dagger}=n^{-1/2}\sum_{i=1}^nW_{n,i}(\tilde{\kappa}_{v,i}^{\dagger}-\kappa_{v,i})\bZ_i\{g_{\alpha,i}(\bgamma^{\dagger},\bbeta^{\dagger})-g_{\alpha,i}(\bgamma^\ast,\bbeta^\ast)\}$, $\mathcal{T}_{n,2}^{\dagger}=n^{1/2}\mathbb{E}[\kappa_{v}\bZ \{g_{\alpha}(\bZ,\bgamma^\ast,\bbeta^{\dagger})-g_{\alpha}(\bZ,\bgamma^\ast,\bbeta^\ast)\}]$, and $\mathcal{T}_{n,3}^{\dagger}=G_n^{\dagger}(\bgamma^{\dagger},\bbeta^{\dagger})-G_n^{\dagger}(\bgamma^\ast,\bbeta^\ast)$. In the following, we will show that $\|\mathcal{T}_{n,i}^{\dagger}\|=o_p(1)$ for $i=\{1,2,3\}$.

We start with showing $\|\mathcal{T}_{n,1}^{\dagger}\|=o_p(1)$. Note that
\begin{equation}
\begin{split}
\|\mathcal{T}_{n,1}^{\dagger}\|=&\|n^{-1/2}\sum_{i=1}^nW_{n,i}(\tilde{\kappa}_{v,i}^{\dagger}-\kappa_{v,i})\bZ_i\{g_{\alpha,i}(\bgamma^{\dagger},\bbeta^{\dagger})-g_{\alpha,i}(\bgamma^\ast,\bbeta^{\dagger})+g_{\alpha,i}(\bgamma^\ast,\bbeta^{\dagger})-g_{\alpha,i}(\bgamma^\ast,\bbeta^\ast)\}\|\\
\le& \|n^{-1/2}\sum_{i=1}^nW_{n,i}(\tilde{\kappa}_{v,i}^{\dagger}-\kappa_{v,i})\bZ_i\{g_{\alpha,i}(\bgamma^{\dagger},\bbeta^{\dagger})-g_{\alpha,i}(\bgamma^\ast,\bbeta^{\dagger})\}\|\\
&+\|n^{-1/2}\sum_{i=1}^nW_{n,i}(\tilde{\kappa}_{v,i}^{\dagger}-\kappa_{v,i})\bZ_i\{g_{\alpha,i}(\bgamma^\ast,\bbeta^{\dagger})-g_{\alpha,i}(\bgamma^\ast,\bbeta^\ast)\}\|.
\end{split}
\label{bootthm2:eq1_1}
\end{equation}

Adapting the similar argument in Equations~\eqref{SQthm2:eq1_2} and \eqref{SQthm2:eq1_3}, we have
\begin{equation}
n^{-1/2}\sum_{i=1}^nW_{n,i}(\tilde{\kappa}_{v,i}^{\dagger}-\kappa_{v,i})\bZ_i\{g_{\alpha,i}(\bgamma^{\dagger},\bbeta^{\dagger})-g_{\alpha,i}(\bgamma^\ast,\bbeta^{\dagger})\}=n^{-1/2}\sum_{i=1}^nW_{n,i}(\tilde{\kappa}_{v,i}^{\dagger}-\kappa_{v,i})\bZ_i\bZ_i^\top(\bgamma^{\dagger}-\bgamma^\ast),
\label{bootthm2:eq1_2}
\end{equation}
and 
\begin{equation}
\|n^{-1/2}\sum_{i=1}^nW_{n,i}(\tilde{\kappa}_{v,i}^{\dagger}-\kappa_{v,i})\bZ_i\{g_{\alpha,i}(\bgamma^\ast,\bbeta^{\dagger})-g_{\alpha,i}(\bgamma^\ast,\bbeta^\ast)\}\|\le n^{-1/2}\sum_{i=1}^nW_{n,i}|\tilde{\kappa}_{v,i}^{\dagger}-\kappa_{v,i}|\|\bZ_i\| C\left(\frac{1}{\alpha}+1\right)\|\bbeta^{\dagger}-{\bbeta}^\ast\|.
\label{bootthm2:eq1_3}
\end{equation}

Coupled with the Equations~\eqref{bootthm2:eq1_1}, \eqref{bootthm2:eq1_2} and ~\eqref{bootthm2:eq1_3},  we have
\begin{equation}
\begin{split}
\|\mathcal{T}_{n,1}^{\dagger}\|\le& \|n^{-1/2}\sum_{i=1}^nW_{n,i}(\tilde{\kappa}_{v,i}^{\dagger}-\kappa_{v,i})\bZ_i\bZ_i^\top(\bgamma^{\dagger}-\bgamma^\ast)\|+n^{-1/2}\sum_{i=1}^nW_{n,i}|\tilde{\kappa}_{v,i}^{\dagger}-\kappa_{v,i}|\|\bZ_i\| C\left(\frac{1}{\alpha}+1\right)\|\bbeta^{\dagger}-{\bbeta}^\ast\|.
\end{split}
\label{bootthm2:eq1}
\end{equation}

Since  $\sup_{i}|\tilde{\kappa}_{v,i}^{\dagger}-\kappa_{v,i}|=o_p(n^{-1/4})$, $\|\bbeta^{\dagger}-\bbeta^\ast\|=o_p(n^{-1/4})$ from the result in  ~\ref{bootbeta}, and
$\|\bgamma^{\dagger}-\bgamma^\ast\|=o_p(n^{-1/4})$, Equation~\eqref{bootthm2:eq1} implies that $\|\mathcal{T}_{n,1}^{\dagger}\|=o_p(1)$.

Note that $\|\bbeta^{\dagger}-\bbeta^\ast\|=o_p(n^{-1/4})$. The proof of $\|\mathcal{T}_{n,2}^{\dagger}\|=o_p(1)$ is identical to the proof of $\|\mathcal{T}_{n,2}\|=o_p(1)$ using $\bbeta^{\dagger}$ in place of $\hat\bbeta$ in the proof of $\|\mathcal{T}_{n,2}\|=o_p(1)$ in Theorem~\ref{SQthm2}. Moreover, adapting the arguments for showing  $\|\mathcal{T}_{n,3}\|=o_p(1)$ in Theorem~\ref{SQthm2} with Poissonization employed to remove dependence among $W_{n,i}$ \cite{Van1997}, we can get $\|\mathcal{T}_{n,3}^{\dagger}\|=o_p(1)$.

From Equation \eqref{bootthm2:eq0}, we have the following equation 
\[
n^{1/2}\mathbb{E}(\kappa_{v}\bZ\bZ^\top)(\bgamma^{\dagger}-\bgamma^\ast)=-n^{-1/2}\sum_{i=1}^nW_{n,i}\tilde{\kappa}_{v,i}^{\dagger}\bZ_i g_{\alpha,i}(\bgamma^\ast,\bbeta^\ast)+o_p(1).
\]

Coupled with Lemma \ref{SQlemma3} and Equation~\eqref{SQthm2:eq4}, we have 
\begin{equation}
n^{1/2}\mathbb{E}\{I(D_1>D_0)\bZ\bZ^\top\}(\bgamma^{\dagger}-\hat\bgamma)=-\Big\{n^{-1/2}\sum_{i=1}^nW_{n,i}\bPsi_i(\alpha)-n^{-1/2}\sum_{i=1}^n\bPsi_i(\alpha)\Big\}+o_p(1).
\label{bootthm2:eq4}
\end{equation}

Following the similar arguments in \ref{bootbeta} and \cite{Wei2021}, since $\{\bPsi_i(\alpha):\alpha\in(0,1)\}$ is a Donsker class from the definition of $\bPsi_i(\alpha)$, by bootstrap consistency for Donsker classes (e.g. Theorem 2.6 of \cite{Kosorok2008}), it implies that the conditional distribution of $n^{-1/2}\sum_{i=1}^nW_{n,i}\bPsi_i(\alpha)-n^{-1/2}\sum_{i=1}^n\bPsi_i(\alpha)$ given the observed data is asymptotically equivalent to the distribution of\\ $n^{-1/2}\sum_{i=1}^n\bPsi_i(\alpha)-n^{1/2}\mathbb{E}\{\bPsi_i(\alpha)\}$. Note that $\mathbb{E}\{\bPsi_i(\alpha)\}=0$ and $\|n^{-1/2}\sum_{i=1}^n\bPsi_i(\alpha)\|=O_p(1)$, and the distribution of $n^{-1/2}\sum_{i=1}^n\bPsi_i(\alpha)-n^{1/2}\mathbb{E}\{\bPsi_i(\alpha)\}$ is asymptotically equivalent to $N(\bzero, \bJ_2^{-1}\bOmega_2\bJ_2^{-1})$ given the result of  Theorem \ref{SQthm2}. We then obtained that $n^{1/2}\{\bgamma^{\dagger}-\hat{\bgamma}\}$ given the observed data is asymptotically equivalent to $N(\bzero, \bJ_2^{-1}\bOmega_2\bJ_2^{-1})$.

\section{Proof of Lemmas~\ref{SQlemma1}--\ref{SQlemma4}}
\label{appendix:C}
\subsection{Proof of Lemmas~\ref{SQlemma1}--\ref{SQlemma2}}
\begin{proof}
By Condition (C6), replacing $\tilde{\kappa}_{v,i}$ with $\hat{\kappa}_{v,i}$ in $n^{-1/2}\sum_{i=1}^n\tilde{\kappa}_{v,i}\bZ_ig_{\alpha,i}(\bgamma^\ast,\bbeta^\ast)$ only leads to a difference of $o_p(1)$. Note that $\bfm_2(Y_i,\bZ_i,\alpha)=\bZ_ig_{\alpha,i}(\bgamma^\ast,\bbeta^\ast)$,  we can obtain the following equation 
\begin{equation}
\begin{split}
&n^{-1/2}\sum_{i=1}^n\tilde{\kappa}_{v,i}\bZ_ig_{\alpha,i}(\bgamma^\ast,\bbeta^\ast)\\
=&n^{-1/2}\sum_{i=1}^n\bfm_2(Y_i,\bZ_i,\alpha)\left(1-\frac{D_i(1-\hat v_i)}{1-\hat{\pi}_{i}}-\frac{(1-D_i)\hat v_i}{\hat{\pi}_{i}}\right)+o_p(1)\\
=&\frac{1}{\sqrt{n}}\sum_{i=1}^n\bfm_2(Y_i,\bZ_i,\alpha)\cdot\left(1-\frac{D_i(1-\hat v_i)}{1-\pi_{i}}-\frac{(1-D_i)\hat v_i}{\pi_{i}}\right)\\
&\quad+\frac{1}{\sqrt{n}}\sum_{i=1}^n\bfm_2(Y_i,\bZ_i,\alpha)\cdot\left(\frac{(1-D_i)\cdot v_{i}}{\pi_{i}^2}-\frac{D_i\cdot (1-v_{i})}{(1-\pi_{i})^2}\right)\cdot(\hat\pi_i-\pi_{i})\\
&\quad+\bR_{n,1}+\bR_{n,2}+o_p(1),
\end{split}
\label{SQlemma1:eq1}
\end{equation}
where 
$$\bR_{n,1}=\frac{1}{n}\sum_{i=1}^n\bfm_2(Y_i,\bZ_i,\alpha)\cdot\left\{\frac{(1-D_i)}{\pi_{i}\hat\pi_i}-\frac{D_i}{(1-\pi_{i})(1-\hat\pi_i)}\right\}\cdot n^{1/4}(\hat\pi_i-\pi_{i})\cdot n^{1/4}(\hat v_i-v_{i}),$$ and $$\bR_{n,2}=\frac{1}{n}\sum_{i=1}^n\bfm_2(Y_i,\bZ_i,\alpha)\cdot\left\{\frac{(1-D_i)\cdot v_{i}}{\pi_{i}^2\hat\pi_i}-\frac{D_i\cdot(1-v_{i})}{(1-\pi_{i})^2(1-\hat\pi_i)}\right\}\cdot n^{1/2}(\hat\pi_i-\pi_{i})^2.$$

Applying Lemma B.3 of \cite{Newey1994} with Condition (C8), we have $\sup_i|\hat\pi_i-\pi_i|=o(n^{-1/4})$ and $\sup_i|\hat v_i-v_i|=o(n^{-1/4})$. Thus,
\[
\begin{split}
\|\bR_{n,1}\|=&\Big\|\frac{1}{n}\sum_{i=1}^n\bfm_2(Y_i,\bZ_i,\alpha)\cdot\left\{\frac{(1-D_i)}{\pi_{i}\hat\pi_i}-\frac{D_i}{(1-\pi_{i})(1-\hat\pi_i)}\right\}\cdot n^{1/4}(\hat\pi_i-\pi_{i})\cdot n^{1/4}(\hat v_i-v_{i})\Big\|\\
\le&n^{1/4}\sup_i|\hat\pi_i-\pi_i|n^{1/4}\sup_i|\hat v_i-v_i|\frac{1}{n}\Big\|\sum_{i=1}^n\bfm_2(Y_i,\bZ_i,\alpha)\cdot\left\{\frac{(1-D_i)}{\pi_{i}\hat\pi_i}-\frac{D_i}{(1-\pi_{i})(1-\hat\pi_i)}\right\}\Big\|\\
&= o_p(1).
\end{split}
\]
 
Similarly, we can show that $\|\bR_{n,2}\|=o_p(1)$.

Given the fact $V$ is a binary variable,  the assumptions required by Theorem 4.2 of \cite{Newey1994} are ensured  by  Conditions (C1), (C7) and (C8). From there, we have
\begin{equation}
\begin{split}
&\frac{1}{\sqrt{n}}\sum_{i=1}^n\bfm_2(Y_i,\bZ_i,\alpha)\cdot\left\{\frac{(1-D_i)\cdot v_{i}}{\pi_{i}^2}-\frac{D_i\cdot (1-v_{i})}{(1-\pi_{i})^2}\right\}\cdot(\hat\pi_i-\pi_{i})\\
&\quad=\frac{1}{\sqrt{n}}\sum_{i=1}^n\bH_2(\bX_i,\alpha)\cdot(V_i-\pi_i)+o_{p}(1),
\label{SQlemma1:eq2}
\end{split}
\end{equation}

and 

\begin{equation}
\begin{split}
&\frac{1}{\sqrt{n}}\sum_{i=1}^n\bfm_2(Y_i,\bZ_i,\alpha)\cdot\left\{1-\frac{D_i(1-\hat v_i)}{1-\pi_{i}}-\frac{(1-D_i)\hat v_i}{\pi_{i}}\right\}\\
&\quad=\frac{1}{\sqrt{n}}\sum_{i=1}^n\bfm_2(Y_i,\bZ_i,\alpha)\cdot\left\{1-\frac{D_i(1-V_i)}{1-\pi_{i}}-\frac{(1-D_i) V_i}{\pi_{i}}\right\}+o_{p}(1),
\label{SQlemma1:eq3}
\end{split}
\end{equation}
as desired.

 In the case of discrete $\bX$ (with finitely many values),  the only difference is that  $\hat{\pi}(x)$ will be  the empirical estimate of $E(V|X=x)$, and in this case we have $\sup_i|\hat{\pi}(x)-\pi(x)|=O_p(n^{-1/2})$. Thus, $\|\bR_{n,j}\|=o_p(1),\, j=1,2$ 
 and Equation~\eqref{SQlemma1:eq3} still holds. To prove Lemma \ref{SQlemma1}, we only need to prove Equation~\eqref{SQlemma1:eq2} with the empirical estimator of 
$\pi(x)=E(V|X=x)$. Let $L$ be the size of the sample space for $\bX$, and $n_l$ be the number of subjects in cell $l$ of $\bX$, $(l=1, \cdots , L)$. Let $x_l$ be the value of  $\bX$ in cell $l$, and $\pi_0^l$ be the expected value of $V$ given $\bX = x_l$.  Then we have
\begin{equation}
\begin{split}
&\frac{1}{\sqrt{n}}\sum_{i=1}^n\bfm_2(Y_i,\bZ_i,\alpha)\cdot\left\{\frac{(1-D_i)\cdot v_{i}}{\pi_{i}^2}-\frac{D_i\cdot (1-v_{i})}{(1-\pi_{i})^2}\right\}\cdot(\hat\pi_i-\pi_{i})\\
&\quad=\sum_{l=1}^L\Big(\frac{1}{\sqrt{n}}\sum_{i_l=1}^{n_l}Z_{i_l}-\pi_0^l\Big)\cdot\Big(\frac{1}{n_l}\sum_{i_l=1}^{n_l}\bfm_2(Y_{i_l},\bZ_{i_l},\alpha)\cdot\left\{\frac{(1-D_{i_l})\cdot v_{i_l}}{(\pi_0^l)^2}-\frac{D_{i_l}\cdot (1-v_{i_l})}{(1-\pi_0^l)^2}\right\}\Big).
\label{SQlemma1:eq4}
\end{split}
\end{equation}

By Lemma 4.3 in \cite{NF1994} we have 
\begin{equation}
\frac{1}{n_l}\sum_{i_l=1}^{n_l}\bfm_2(Y_{i_l},\bZ_{i_l},\alpha)\cdot\left\{\frac{(1-D_{i_l})\cdot v_{i_l}}{(\pi_0^l)^2}-\frac{D_{i_l}\cdot (1-v_{i_l})}{(1-\pi_0^l)^2}\right\}=\bH_2(x_l,\alpha)+o_p(1).
\label{SQlemma1:eq5}
\end{equation}

Thus, Equation~\eqref{SQlemma1:eq2} follows from Equations~\eqref{SQlemma1:eq4} and ~\eqref{SQlemma1:eq5}.

The proof of Lemma \ref{SQlemma2} is similar to the proof of Lemma \ref{SQlemma1}   using $\{\alpha-I(Y_i<\bZ_i^\top\bbeta^\ast)\}$ in place of $g_{\alpha,i}(\bgamma^\ast,\bbeta^\ast)$ in the proof of Lemma \ref{SQlemma1}, and is hence omitted.

\end{proof}

\subsection{Proof of Lemmas~\ref{SQlemma3}--\ref{SQlemma4}}

\begin{proof}

The proofs of Lemmas~\ref{SQlemma3}--\ref{SQlemma4} are adapted from the proof of Lemmas  \ref{SQlemma1} and \ref{SQlemma2} and the justification of the standard nonparametric bootstrap inference procedure in \cite{Wei2021}. We will prove them in detail here for completeness.

By Condition (C6), replacing $\tilde{\kappa}_{v,i}^{\dagger}$ with ${\kappa}_{v,i}^{\dagger}$ in $n^{-1/2}\sum_{i=1}^nW_{n,i}\tilde{\kappa}_{v,i}^{\dagger}\bZ_ig_{\alpha,i}(\bgamma^\ast,\bbeta^\ast)$ only leads to a difference of $o_p(1)$. Note that $\bfm_2(Y_i,\bZ_i,\alpha)=\bZ_ig_{\alpha,i}(\bgamma^\ast,\bbeta^\ast)$,  we can obtain the following equation 
\begin{equation}
\begin{split}
&n^{-1/2}\sum_{i=1}^nW_{n,i}\tilde{\kappa}_{v,i}^{\dagger}\bZ_ig_{\alpha,i}(\bgamma^\ast,\bbeta^\ast)\\
=&n^{-1/2}\sum_{i=1}^nW_{n,i}\bfm_2(Y_i,\bZ_i,\alpha)\left(1-\frac{D_i(1-v_i^{\dagger})}{1-\pi^{\dagger}_{i}}-\frac{(1-D_i)v_i^{\dagger}}{\pi^{\dagger}_{i}}\right)+o_p(1)\\
=&\frac{1}{\sqrt{n}}\sum_{i=1}^nW_{n,i}\bfm_2(Y_i,\bZ_i,\alpha)\cdot\left(1-\frac{D_i(1- v_i^{\dagger})}{1-\pi_{i}}-\frac{(1-D_i)v_i^{\dagger}}{\pi_{i}}\right)\\
&\quad+\frac{1}{\sqrt{n}}\sum_{i=1}^nW_{n,i}\bfm_2(Y_i,\bZ_i,\alpha)\cdot\left(\frac{(1-D_i)\cdot v_{i}}{\pi_{i}^2}-\frac{D_i\cdot (1-v_{i})}{(1-\pi_{i})^2}\right)\cdot(\pi_i^\dagger-\pi_{i})\\
&\quad+\bR_{n,1}^{\dagger}+\bR_{n,2}^{\dagger}+o_p(1),
\end{split}
\label{SQlemma3:eq1}
\end{equation}
where 
$$\bR_{n,1}^{\dagger}=\frac{1}{n}\sum_{i=1}^nW_{n,i}\bfm_2(Y_i,\bZ_i,\alpha)\cdot\left\{\frac{(1-D_i)}{\pi_{i}\pi_i^{\dagger}}-\frac{D_i}{(1-\pi_{i})(1-\pi_i^{\dagger})}\right\}\cdot n^{1/4}(\pi_i^{\dagger}-\pi_{i})\cdot n^{1/4}( v_i^{\dagger}-v_{i}),$$ and $$\bR_{n,2}^{\dagger}=\frac{1}{n}\sum_{i=1}^nW_{n,i}\bfm_2(Y_i,\bZ_i,\alpha)\cdot\left\{\frac{(1-D_i)\cdot v_{i}}{\pi_{i}^2\pi_i^{\dagger}}-\frac{D_i\cdot(1-v_{i})}{(1-\pi_{i})^2(1-\pi_i^{\dagger})}\right\}\cdot n^{1/2}(\pi_i^{\dagger}-\pi_{i})^2.$$

Since we have $\sup_i|\pi_i^{\dagger}-\pi_i|=o(n^{-1/4})$ and $\sup_i|v_i^{\dagger}-v_i|=o(n^{-1/4})$, with the similar argument in Lemma \ref{SQlemma1}, we have $\|\bR_{n,1}^{\dagger}\|=o_p(1)$. and $\|\bR_{n,2}^{\dagger}\|=o_p(1)$.

Combining the technique of \cite{Newey1994}, we have 
\begin{equation}
\begin{split}
&\frac{1}{\sqrt{n}}\sum_{i=1}^n W_{n,i}\bfm_2(Y_i,\bZ_i,\alpha)\cdot\left\{\frac{(1-D_i)\cdot v_{i}}{\pi_{i}^2}-\frac{D_i\cdot (1-v_{i})}{(1-\pi_{i})^2}\right\}\cdot(\pi_i^{\dagger}-\pi_{i})\\
&\quad=\frac{1}{\sqrt{n}}\sum_{i=1}^n W_{n,i}\bH_2(\bX_i,\alpha)\cdot(V_i-\pi_i)+o_{p}(1),
\label{SQlemma3:eq2}
\end{split}
\end{equation}

and 

\begin{equation}
\begin{split}
&\frac{1}{\sqrt{n}}\sum_{i=1}^n W_{n,i}\bfm_2(Y_i,\bZ_i,\alpha)\cdot\left\{1-\frac{D_i(1-v_i^{\dagger})}{1-\pi_{i}}-\frac{(1-D_i)v_i^{\dagger}}{\pi_{i}}\right\}\\
&\quad=\frac{1}{\sqrt{n}}\sum_{i=1}^n W_{n,i}\bfm_2(Y_i,\bZ_i,\alpha)\cdot\left\{1-\frac{D_i(1-V_i)}{1-\pi_{i}}-\frac{(1-D_i) V_i}{\pi_{i}}\right\}+o_{p}(1),
\label{SQlemma3:eq3}
\end{split}
\end{equation}
as desired. 

We can also prove Lemma \ref{SQlemma3} in the case of discrete $\bX$ by adapting the similar arguments in Lemma \ref{SQlemma1} and the proof of the above.

The proof of Lemma \ref{SQlemma4} is identical to the proof of Lemma \ref{SQlemma3}   using $\{\alpha-I(Y_i<\bZ_i^\top\bbeta^\ast)\}$ in place of $g_{\alpha,i}(\bgamma^\ast,\bbeta^\ast)$ in the proof of Lemma \ref{SQlemma3}, and is hence omitted.

\end{proof}

\section{Simulation with Discrete $\bX$}
 We now consider a new simulation scenario, which differs from the scenario in Section \ref{simu} of the paper in that $X_1$ and $X_2$ are both discrete and generated from $Bernoulli(0.5)$. Results with $n=\{ 500,3000\}$ for $\alpha = \{0.1,0.2,0.3,0.4,0.5\}$ are presented in Table~\ref{SQsimudiscrete:t1}. 
The performance of the proposed method remains solid in this case and the estimated parameters of interest are close to their true underlying values. Besides, the bootstrap-based variance estimates agree well with the empirical variances. The empirical coverage probabilities of 95\% confidence intervals are close to the nominal level. When the sample size  increases to $n=3000$, the bias of the proposed method further diminishes.

\begin{table}[!t]
	\fontsize{8}{9}\selectfont
	\centering
	\caption{Comparisons among estimators of $\beta_{1}(\alpha)$ and $\gamma_{1}(\alpha)$ from the proposed two-stage method, oracle method and the naive method in simulation experiment on the case of discrete $\bX$ with $n=\{500,3000\}$ and $\alpha=\{0.1,0.2,\ldots, 0.5\}$. Bias, Emp var, Boot var and Cov 95 stand for average bias of the estimated coefficients, empirical variance, average variance estimates, coverage probabilities of the 95\% confidence intervals.  }
	\label{SQsimudiscrete:t1}
	 \begin{threeparttable}
	
\begin{tabular}{|c|c|l|cc|cc|cc|}
    \hline
\multirow{2}{*}{$\alpha$} &\multirow{2}{*}{$n$}&&\multicolumn{2}{c|}{Oracle method}  & \multicolumn{2}{c|}{Proposed method} &\multicolumn{2}{c|}{Naive method} \\
 \cline{3-9}
&  & &$\beta_{1}(\alpha)$&$\gamma_{1}(\alpha)$&$\beta_{1}(\alpha)$&$\gamma_{1}(\alpha)$&$\beta_{1}(\alpha)$&$\gamma_{1}(\alpha)$\\
\hline
 \multirow{8}{*}{0.1}&\multirow{4}{*}{500}&Bias&-0.013&-0.027&-0.039&-0.053&-0.255&-0.268\\
     \cline{3-9}
     && Emp var&0.131&0.284&0.134&0.272&0.082&0.170\\   
      \cline{3-9}
     && Boot var&0.156&0.264&0.153&0.256&0.094&0.162\\   
       \cline{3-9}
      && Cov 95&0.956&0.924&0.949&0.930&0.857&0.861\\ 
       \cline{2-9}
       &{\multirow{4}{*}{3000}}&Bias&0.000&-0.003&-0.017&-0.018&-0.245&-0.251\\
     \cline{3-9}&& Emp var&0.026&0.049&0.026&0.048&0.016&0.031\\   
      \cline{3-9}&& Boot var&0.025&0.046&0.025&0.045&0.015&0.028\\   
      \cline{3-9}&& Cov 95&0.933&0.950&0.935&0.946&0.482&0.652\\  
      \hline
      \multirow{8}{*}{0.2}&\multirow{4}{*}{500}&Bias&-0.007&0.016&-0.029&-0.044&-0.245&-0.256\\
     \cline{3-9}&& Emp var&0.060&0.132&0.064&0.131&0.031&0.078\\   
      \cline{3-9}&& Boot var&0.069&0.131&0.071&0.129&0.035&0.076\\   
      \cline{3-9}&& Cov 95&0.955&0.947&0.949&0.933&0.721&0.822\\
      \cline{2-9}
       &{\multirow{4}{*}{3000}}&Bias&-0.001&-0.002&-0.015&-0.017&-0.243&-0.249\\
     \cline{3-9}&& Emp var&0.011&0.025&0.012&0.025&0.005&0.014\\   
     \cline{3-9}&& Boot var&0.011&0.022&0.011&0.022&0.006&0.013\\   
     \cline{3-9}&& Cov 95&0.936&0.936&0.941&0.930&0.10&0.395\\   
      \hline
  \multirow{8}{*}{0.3}&\multirow{4}{*}{500}&Bias&0.000&-0.012&-0.019&-0.037&-0.239&-0.250\\
     \cline{3-9}&& Emp var&0.036&0.084&0.038&0.085&0.015&0.047\\   
     \cline{3-9}&& Boot var&0.041&0.084&0.044&0.084&0.017&0.046\\   
      \cline{3-9}&& Cov 95&0.958&0.949&0.958&0.942&0.533&0.757\\  
      \cline{2-9}
       &{\multirow{4}{*}{3000}}&Bias&-0.001&-0.001&-0.015&-0.016&-0.240&-0.246\\
     \cline{3-9}&& Emp var&0.006&0.016&0.007&0.016&0.002&0.008\\   
     \cline{3-9}&& Boot var&0.006&0.014&0.007&0.014&0.003&0.008\\   
      \cline{3-9}&& Cov 95&0.947&0.933&0.949&0.928&0.003&0.210\\     
      \hline      
      
  \multirow{8}{*}{0.4}&\multirow{4}{*}{500}&Bias&-0.001&-0.009&-0.028&-0.034&-0.237&-0.248\\
     \cline{3-9}&& Emp var&0.023&0.060&0.025&0.061&0.009&0.032\\   
      \cline{3-9}&& Boot var&0.027&0.060&0.029&0.061&0.009&0.031\\   
      \cline{3-9}&& Cov 95&0.955&0.948&0.956&0.937&0.323&0.694\\  
      \cline{2-9}
       &{\multirow{4}{*}{3000}}&Bias&0.001&-0.001&-0.013&-0.016&-0.238&-0.245\\
     \cline{3-9}&& Emp var&0.004&0.011&0.005&0.011&0.001&0.006\\   
      \cline{3-9}&& Boot var&0.004&0.010&0.005&0.010&0.001&0.005\\   
      \cline{3-9}&& Cov 95&0.951&0.937&0.952&0.933&0.000&0.076\\       
      \hline

      \multirow{8}{*}{0.5}&\multirow{4}{*}{500}&Bias&-0.007&-0.007&-0.033&-0.033&-0.233&-0.244\\
     \cline{3-9}&& Emp var&0.015&0.045&0.017&0.047&0.006&0.023\\   
      \cline{3-9}&& Boot var&0.018&0.045&0.020&0.046&0.006&0.023\\   
      \cline{3-9}&& Cov 95&0.954&0.950&0.947&0.939&0.167&0.615\\  
      \cline{2-9}
       &{\multirow{4}{*}{3000}}&Bias&0.001&0.000&-0.012&-0.015&-0.233&-0.243\\
     \cline{3-9}&& Emp var&0.003&0.008&0.003&0.009&0.001&0.004\\   
     \cline{3-9}&& Boot var&0.003&0.008&0.003&0.008&0.001&0.004\\   
      \cline{3-9}&& Cov 95&0.942&0.942&0.945&0.939&0.000&0.036\\    
      \hline      
\end{tabular}
     \end{threeparttable}
\end{table}

\bibliographystyle{elsarticle-harv} \bibliography{ES}

\end{document}